\documentclass[reprint,amsfonts, amssymb, amsmath,  showkeys,pra, superscriptaddress, twocolumn,longbibliography]{revtex4-2}

\usepackage{xcolor}

\usepackage{float}
\makeatletter
\let\newfloat\newfloat@ltx
\makeatother
\usepackage[english]{babel}
\usepackage[utf8]{inputenc}
\usepackage{graphics}
\usepackage{selinput}
\usepackage[normalem]{ulem}
\usepackage[shortlabels]{enumitem}

\usepackage{braket}
\usepackage{amsthm}
\usepackage{mathtools}
\usepackage{physics}
\usepackage{graphicx}
\usepackage[left=16mm,right=16mm,top=35mm,columnsep=15pt]{geometry} 
\usepackage{adjustbox}
\usepackage{placeins}
\usepackage[T1]{fontenc}
\usepackage{lipsum}
\usepackage{csquotes}
\usepackage{bm}

\usepackage[linesnumbered,ruled,vlined]{algorithm2e}
\SetKwInput{kwInit}{Init}

\def\HC{\mathcal{H}}

\def\LC{\mathcal{L}}

\def\ad{^{\dagger}}

\newcommand{\fsnull}[1]{}
\newcommand{\old}[1]{}

\usepackage[makeroom]{cancel}
\usepackage[toc,page]{appendix}
\usepackage[colorlinks=true,citecolor=blue,linkcolor=magenta]{hyperref}

\usepackage{tikz}
\tikzset{every picture/.style=remember picture}

\usepackage[utf8]{inputenc}
\usepackage{graphicx}
\usepackage{xcolor}
\usepackage{amsmath}
\usepackage{amsthm}
\usepackage{bm}
\usepackage{bbm}
\usepackage{comment}
\usepackage{appendix}
\usepackage{mathdots}
\usepackage{lipsum}
\usepackage{verbatim}
\usepackage{natbib}
\usepackage{nccmath}
\usepackage{amsfonts} 

\newcommand{\dya}[1]{\ket{#1}\!\bra{#1}}
\newcommand{\inprod}[1]{\left\langle #1 \right\rangle}

\newcommand{\poly}{\operatorname{poly}}

\newcommand{\EC}{\mathcal{E}}

\newcommand{\GC}{\mathcal{G}}
\newcommand{\MC}{\mathcal{M}}
\newcommand{\NC}{\mathcal{N}}
\newcommand{\OC}{\mathcal{O}}
\newcommand{\PC}{\mathcal{P}}

\newcommand{\Var}{{\rm Var}}

\newcommand{\supp}{\text{supp}}

\renewcommand{\geq}{\geqslant}
\renewcommand{\leq}{\leqslant}

\DeclareMathOperator*{\argmax}{arg\,max}

\renewcommand{\vec}[1]{\boldsymbol{#1}}  

\newcommand*{\id}{\openone}

\newcommand{\bs}{\textsf{BS}}

\renewcommand{\th}{\theta } 

\def\calA{\mathcal{A}}

\def\calE{\mathcal{E}}

\def\calH{\mathcal{H}}

\def\calM{\mathcal{M}}

\def\C{\mathbb{C}}

\newcommand{\thv}{\vec{\theta}}

\def\be{\begin{equation}}
\def\ee{\end{equation}}
\def\bs{\begin{split}}
\def\e{\end{split}}
\def\ba{\begin{eqnarray}}
\def\bea{\begin{eqnarray}}

\def\tea{\end{eqnarray}}
\def\ea{\end{eqnarray}}
\def\eea{\end{eqnarray}}

\def\h{\vec{h}}

\def\gl{\mathfrak{gl}}
\newcommand{\paran}[1]{\left( #1 \right)}
\def\bbE{\mathbb{E}}

\def\su{\mathfrak{s}\mathfrak{u}}

\def\g{\mathfrak{g}}
\def\u{\mathfrak{u}}
\def\h{\mathfrak{h}}

\def\su{\mathfrak{s}\mathfrak{u}}

\newcommand\mbb[1]{\mathbb{#1}}
\newcommand\mf[1]{\mathfrak{#1}}

\theoremstyle{remark}

\newtheorem{theorem}{Theorem}
\newtheorem{lemma}{Lemma}
\newtheorem{remark}{Remark}
\newtheorem{corollary}{Corollary}

\newtheorem{proposition}{Proposition}

\usepackage{amssymb}
\usepackage{dsfont}

\def\be{\begin{equation}}
\def\te{\end{equation}}
\def\ee{\end{equation}}
\def\ba{\begin{eqnarray}}
\def\bea{\begin{eqnarray}}

\def\tea{\end{eqnarray}}
\def\ea{\end{eqnarray}}
\def\eea{\end{eqnarray}}

\begin{document}

\title{A Lie Algebraic Theory of Barren Plateaus for Deep Parameterized Quantum Circuits}  

\author{Michael Ragone}
\thanks{The first two authors contributed equally.}
\affiliation{Department of Mathematics, University of California Davis, Davis, California 95616, USA}

\author{Bojko N. Bakalov}
\thanks{The first two authors contributed equally.}
\affiliation{Department of Mathematics, North Carolina State University, Raleigh, North Carolina 27695, USA}

\author{Fr\'{e}d\'{e}ric Sauvage}
\affiliation{Theoretical Division, Los Alamos National Laboratory, Los Alamos, New Mexico 87545, USA}

\author{Alexander F. Kemper}
\affiliation{Department of Physics, North Carolina State University, Raleigh, North Carolina 27695, USA}

\author{Carlos Ortiz Marrero}
\affiliation{AI \& Data Analytics Division, Pacific Northwest National Laboratory, Richland, WA 99354, USA}
\affiliation{Department of Electrical \& Computer Engineering, North Carolina State University, Raleigh, North Carolina 27607, USA}

\author{Mart\'{i}n Larocca}
\affiliation{Theoretical Division, Los Alamos National Laboratory, Los Alamos, New Mexico 87545, USA}
\affiliation{Center for Nonlinear Studies, Los Alamos National Laboratory, Los Alamos, New Mexico 87545, USA}

\author{M. Cerezo}
\email{cerezo@lanl.gov}
\affiliation{Information Sciences, Los Alamos National Laboratory, Los Alamos, New Mexico 87545, USA}

\begin{abstract}

Variational quantum computing schemes train a loss function by sending an initial state through a parametrized quantum circuit, and measuring the expectation value of some operator. Despite their promise, the trainability of these algorithms is hindered by barren plateaus (BPs) induced by the expressiveness of the circuit, the entanglement of the input data, the locality of the observable, or the presence of noise. Up to this point, these sources of BPs have been regarded as independent. In this work, we present a general Lie algebraic theory that provides an exact expression for the variance of the loss function of sufficiently deep parametrized quantum circuits, even in the presence of certain noise models. Our results allow us to understand under one framework all aforementioned sources of BPs. This theoretical leap resolves a standing conjecture about a connection between loss concentration and the dimension of the Lie algebra of the circuit's generators.

\end{abstract}

\maketitle

\section{Introduction}

Variational quantum computing schemes, such as variational quantum algorithms~\cite{cerezo2020variationalreview,bharti2021noisy,endo2021hybrid,peruzzo2014variational,farhi2014quantum,hadfield2019quantum} or quantum machine learning models~\cite{schuld2015introduction,biamonte2017quantum,benedetti2019parameterized,havlivcek2019supervised}, share a common structure in which quantum and classical resources are used to solve a given task. In a nutshell, these algorithms send some initial state through a parametrized quantum circuit, and then perform (a polynomial number of) measurements to estimate the expectation value of some observable that encodes the loss function (also called cost function) appropriate for the problem. 
Subsequently, the estimated loss (or its gradient) is fed into a classical optimizer that attempts to update the circuit parameters to minimize the loss.

In the past few years, research has started to point towards seemingly fundamental difficulties to training generic parametrized quantum circuits~\cite{anschuetz2022beyond,bittel2021training,mcclean2018barren}. In particular, one of the main obstacles towards trainability is the presence of Barren Plateaus (BPs)~\cite{mcclean2018barren} in the loss function. In the presence of BPs,  this loss function (and its gradients) exponentially concentrates in parameter space as the size of the problem increases. Therefore, unless an exponential number of measurement shots are employed, the model becomes untrainable, as one does not have enough precision to find a loss-minimizing direction and navigate the loss landscape.  

Due to the tremendous limitations that BPs place on the potential to scale variational quantum computing schemes to large problem sizes, a significant amount of effort has been put forward towards understanding why and when BPs arise. In this context, the presence of BPs (in the absence of noise) has been shown to arise from several disparate
aspects of the variational problem, including
the expressiveness of parametrized quantum circuits (that is, the breadth of unitaries that the parametrized quantum circuit can express)~\cite{mcclean2018barren,holmes2021connecting,marrero2020entanglement,patti2020entanglement,larocca2021diagnosing,friedrich2023quantum,sharma2020trainability,kieferova2021quantum,pesah2020absence,lee2021towards,martin2022barren,grimsley2022adapt,sack2022avoiding}, the locality of the loss function measurement operator $O$ ~\cite{cerezo2020cost,uvarov2020barren,kashif2023impact,khatri2019quantum,uvarov2020variational,leadbeater2021f,cerezo2020variational}, and the entanglement and randomness of the initial state $\rho$~\cite{mcclean2018barren,cerezo2020cost,thanaslip2021subtleties,shaydulin2021importance,abbas2020power,leone2022practical,holmes2021barren}. Hardware noise further exacerbates these issues~\cite{wang2020noise,franca2020limitations,garcia2023effects}. 
Yet, despite our significant understanding of BPs, most of the results in the literature have been derived, or can be applied, only for certain circuit architectures or scenarios. 
Thus, we cannot generalize the lessons learned from one scenario to another, and the different sources of BPs are regarded as independent. In other words, we do not have a unifying holistic theory that can capture the interplay of the various aspects that give rise to BPs.

In this work, we present a general Lie algebraic theory for BPs that can be applied to any deep, unitary, parametrized quantum circuit architecture. Our theory is based on the study of the Lie group and the associated Lie algebra $\g$, which is generated by the parametrized quantum circuit. In turn, this allows us to understand  under one single umbrella all known sources of BPs. Critically, we are able to precisely compute the variance of the loss function, and therefore study its concentration in parameter space, provided that the measured observable $O$ or the circuit's input state $\rho$ is in $i\g$.  Our results encapsulate the known causes of BPs, as we generalize the concepts of circuit expressiveness, initial state entanglement, operator locality, and hardware noise to a unified framework (see Fig.~\ref{fig:Schematic}). Moreover, we also provide a quantifiable definition for a deep quantum circuit, as we present rigorous bounds for the number of layers needed for it to be an approximate design over the Lie group $e^{\g}$, and for how much the variances for an exact and for an approximate design deviate.

\section{Results}

\subsection{Loss function and barren plateaus}

\begin{figure}
    \centering
    \includegraphics[width=.75\columnwidth]{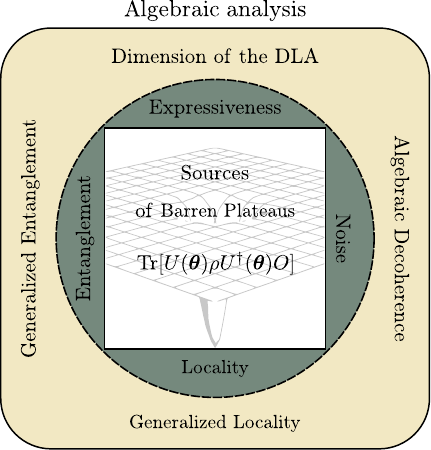}
    \caption{\textbf{Schematic representation of our results.} Many works in the literature have shown that BPs can arise from: the expressiveness of the parametrized circuit, the locality of the measurement operator, the entanglement in the initial state, or due to hardware noise. In the figure, we schematically depict these sources as being separate, since their study is usually performed in restricted scenarios, thus limiting our ability to interconnect them.  In this work, we present a unified Lie algebraic theory that can be used to study BPs in deep parametrized quantum circuits under SPAM and coherent noise. Our theorems allow us to understand under a single unified framework all sources of BPs by leveraging concepts such as the generalized locality and entanglement and algebraic decoherence. }
    \label{fig:Schematic}
\end{figure}

In what follows, we will consider problems where an $n$-qubit state $\rho$, acting on a Hilbert space $\HC =(\mathbb{C}^2)^{\otimes n}$, is sent through a parametrized quantum circuit $U(\thv)$ of the form
\begin{equation}\label{eq:PQC}
    U(\thv)=\prod_{l=1}^L e^{i H_l \th_l}\,.
\end{equation}
Here, $\thv=(\theta_1,\theta_2,\ldots)$ is a set of trainable real-valued parameters, and $H_l\in i\mf{u}(2^n)$ are Hermitian operators that we collect in a set of generators $\GC=\{H_1,H_2,\ldots\}$.  
At the output of the circuit, one measures some Hermitian observable $O\in i\mf{u}(2^n)$, such that $\norm{O}_2^2\leq 2^n$, leading to the loss function,
 \begin{equation}\label{eq:loss}
        \ell_{\thv}(\rho, O) = \Tr[ U(\thv) \rho U^\dagger(\thv) O]\,.
 \end{equation}
 Depending on the algorithm at hand, the previous procedure might be repeated for different initial states (e.g., input states coming from some training set), or measurement operators (e.g., measuring different non-commuting observables), and the set of such expectation values can be combined into a more complicated loss function. To account for all of those scenarios, we will simply study a quantity as that in Eq.~\eqref{eq:loss}, as we know that if $\ell_{\thv}(\rho, O)$ concentrates, then so will any function computed from it. 

Finally, we will consider SPAM errors --- which we model by adding completely positive and trace preserving channels $\NC_B$ and $\NC_A$ acting before and after the circuit, respectively --- and coherent errors, which we represent as uncontrolled unitary gates interleaved with the parametrized gates. When these sources of noise are present, we write the noisy loss function as
 \begin{align}\label{eq:loss-noisy}
   \widetilde{\ell}_{\thv}(\rho, O) &= \Tr\left[ \NC_A\left(\widetilde{U}(\thv) \NC_B(\rho) \widetilde{U}^\dagger(\thv)\right)O\right]\,,
 \end{align}
where $\widetilde{U}(\thv)=\prod_{l=1}^L e^{-i \alpha_l K_l} e^{i H_l \theta_l}$ for fixed (not tunable) real values $\alpha_l$ and Hermitian operators $K_l$.

In what follows, we will study how much the loss $\ell_{\thv}(\rho, O)$ (or its noisy counterpart $\widetilde{\ell}_{\thv}(\rho, O)$) changes as the parameters $\thv$ vary. In particular, we want to characterize the variance of the loss over the parameter landscape, 
\begin{equation} \label{eqn:variance}
\Var_{\thv} [\ell_{\thv}(\rho, O)] = \bbE_{\thv}  [\ell_{\thv}(\rho, O)^2] - \paran{\bbE_{\thv}   [\ell_{\thv}(\rho, O)]}^2\,,
\end{equation}
and we will say that the loss exhibits a BP if the variance vanishes exponentially with the system size, i.e., if $\Var_{\thv} [\ell_{\thv}(\rho, O)]\in\OC(1/b^n)$, for some $b>1$. It is important to note that in the literature it is common to study the presence of BPs by analyzing the concentration of the partial derivatives of the loss, i.e., the scaling of $\Var_{\thv} [ \frac{\partial\ell_{\thv}(\rho, O)}{\partial\theta_\mu}]$. However, since loss concentration implies partial derivative concentration~\cite{arrasmith2021equivalence}, in what follows we exclusively focus on the former. Moreover, our results on the variance of the loss function can be modified in a straightforward way to obtain the variance of the partial derivatives.

\subsection{Dynamical Lie algebra}\label{sec:DLA}

Computing the variance of the loss function requires evaluating averages over the parameter landscape. When the circuit is sufficiently deep,  we can leverage Lie algebraic tools to perform such calculations. To this end, we define the so-called dynamical Lie algebra (DLA)~\cite{zeier2011symmetry,dalessandro2010introduction,zimboras2015symmetry} of the parametrized quantum circuit (we will assume a circuit with no coherent errors for the remainder of this section). The DLA is defined as the Lie closure of the circuit's generators,
\begin{equation}\label{eq:dla}
    \g=\langle i\GC\rangle_{{\rm Lie}} \,,
\end{equation}
which is the (closed under commutation) subspace of $\mf{u}(2^n)$ spanned by the nested commutators of the generators $i\GC$. Since $\g$ is a subalgebra of the Lie algebra $\u(2^n)$ of skew-Hermitian operators acting on $\HC$, it is a reductive Lie algebra (see e.g.\ \cite{knapp2013lie}, Chapter IV). This means that we can express $\g$ as a direct sum of commuting ideals,
\begin{equation}\label{eq:reductive}
        \g = \g_1 \oplus \dots\oplus \g_{k-1} \oplus \g_{k}\,,
\end{equation}  
where $\g_j$ are simple Lie algebras for $j=1,\ldots,k-1$ and $\g_{k}$ is abelian (so $\g_k$ is the center of $\g$ and $[\g,\g]=\g_1 \oplus \dots\oplus \g_{k-1}$ is semisimple). In the Methods we provide additional intuition as to why the DLA takes the form in Eq.~\eqref{eq:reductive}. The importance of the DLA comes from the fact that it quantifies the ultimate expressiveness of the parametrized quantum circuit: we have $U(\thv)\in G=e^{\mf{g}}$ for all $\thv$.

In the following, we assume that the parametrized circuit is deep enough so that it forms a  $2$-design on each of the simple or abelian components $G_j=e^{\mf{g}_j}$ ($j=1,\ldots,k)$ (i.e.,  the first two moments of the distribution of unitaries obtained from the circuits, match those of the Haar measure over  each group $G_j$~\cite{dankert2009exact}), which will allow us to compute the variance via an explicit integration over the Lie group. Our results are readily useful to analyze the variance of circuits which form approximate, rather than exact, designs. In particular, in the Methods and the Supplemental Information, we present and prove Theorems~\ref{th:design-layers} and~\ref{th:variance-layers}, which respectively set bounds on the necessary number of layers for the circuit to become an $\epsilon$-approximate $2$-design over $G$, and bound the difference between the variance of a $2$-design circuit and that of an $L$-layered circuit that does not necessarily form a $2$-design over $G$. 
Then, for any function $f(U(\thv))$ that is a polynomial of order (up to) two in the matrix elements of $U(\thv)$ and its conjugate transpose, we will compute averages over the parameter landscape as $\bbE_{\thv}[f(U(\thv))]=\prod_{j=1}^k\int_{G_j} d\mu_j f(U(\thv))$\,,
where $d\mu_j$ is the normalized, left- and right-invariant Haar measure on $G_j$.
Such integrals can be evaluated using Weingarten Calculus (see, e.g., \cite{collins2022weingarten}), and we present a simple approach based on invariant theory of Lie algebras.

Finally, our results use a generalized notion of purity with respect to an arbitrary operator subalgebra $\mf{g}\subseteq \mf{u}(2^n)$. The $\mf{g}$-purity of a Hermitian
operator $H\in i \mf{u}(2^n)$ is defined as~\cite{somma2004nature,somma2005quantum}
\begin{equation}\label{eq:g-purity}
\PC_{\g}(H)=\Tr[H_\g^2]=\sum_{j=1}^{\dim( \g)} \abs{\Tr[B_j\ad H]}^2\,,
\end{equation}
where $H_\g$ denotes the orthogonal projection of $H$ into  $\g_{\mathbb{C}}=\mathrm{span}_{\mbb{C}} \mf{g}$ (the complexification of $\mf{g}$)
and $\{B_j\}_{j=1}^{\dim(\g)}$ an orthonormal basis (over $\mathbb C$) for $\g_{\mathbb{C}}$ 
with respect to the Hilbert--Schmidt inner product $\inprod{A,B} = \Tr [A^\dagger B]$~\cite{knapp2013lie}.
Note that $H_\g\in i\g$ and the basis vectors $B_j$ are chosen from $i\g$;
in that case, $\Tr[B_j\ad H] = \Tr[B_j H] \in \mathbb{R}$.
In the Supplemental Information, we provide some intuition about why the $\g$-purity arises as the natural quantity in our results.

\subsection{Main result}

In this section, we present an exact formula for the variance of the loss function. To simplify the presentation, we will first restrict ourselves to noiseless circuits and incorporate noise as we proceed. Our main theorem exploits the reductiveness of the DLA: when either $\rho\in i\mf{g}$ or $O\in i\mf{g}$, the loss (and its variance) can be expanded into individual terms arising from each ideal $\g_j\subseteq\g$ (see the Methods and Supplemental Information for additional details). Equipped with this result, we can prove the following theorem:

\begin{theorem}\label{thm:variance-reductive}
    Suppose that $O\in i\g$ or $\rho \in i\g$, where the DLA $\g$ is as in Eq.~\eqref{eq:reductive}.
    Then the mean of the loss function vanishes for the semisimple component $\g_1 \oplus \dots\oplus \g_{k-1}$ and leaves only abelian     contributions:   
    \begin{equation}\label{eq:exq-main}
       \bbE_{\thv}[\ell_{\thv}(\rho,O)] = \Tr [\rho_{\g_k} O_{\g_k}] \,.
    \end{equation}
Conversely, the variance of the loss function vanishes for the center $\g_k$ and leaves only simple contributions:
\begin{align}\label{eq:var-main}
    \Var_{\thv} [\ell_{\thv} (\rho,O)] =\sum_{j=1}^{k-1} \frac{\PC_{\g_j}(\rho) \PC_{\g_j}(O)}{\dim (\g_j)} \,.
\end{align}       
\end{theorem}

Let us now consider the implications of Theorem~\ref{thm:variance-reductive}. As we can see from  Eq.~\eqref{eq:exq-main},  the mean of the loss is solely determined by the projections of $\rho$ and $O$ onto the center $\g_k$ of $\mf{g}$. As such, if $\mf{g}$ is centerless, then $\bbE_{\thv}[\ell_{\thv}(\rho,O)]=0$. Conversely,  if $\g$ is abelian and if $\rho$ or $O$ commutes with $\mf{g}$ (a slightly more general condition than belonging to $\mf{g}$), then the loss landscape is completely flat. 

Next, we turn our attention to Eq.~\eqref{eq:var-main}. This equation shows that each term in the variance of the loss function depends on three quantities: the dimension of $\g_j$, and the $\g_j$-purities of $\rho$ and $O$. To understand what each of these terms means, let us focus on the case where $\g$ is simple. Then Eq.~\eqref{eq:var-main} has a single term, 
\begin{equation}
    \Var_{\thv}[\ell_{\thv}(\rho,O)] = \frac{\PC_{\g}(\rho) \PC_{\g}(O)}{\dim (\g)} \,. 
\end{equation} 
The following result, proved in the Supplemental Information, illustrates the fact that BPs can arise from three sources.
\begin{corollary} \label{cor:concentration}
Let $\norm{O}_2^2\leq 2^n$. If either $\dim(\g)$, $1/\PC_{\g}(\rho)$ or $1/\PC_{\g}(O)$ is in $\Omega(b^n)$ with $b> 2$, then the loss has a BP. 
\end{corollary}

In what follows, we analyze these three causes of BPs.

\begin{figure*}[ht]
    \centering
    \includegraphics[width=.9\linewidth]{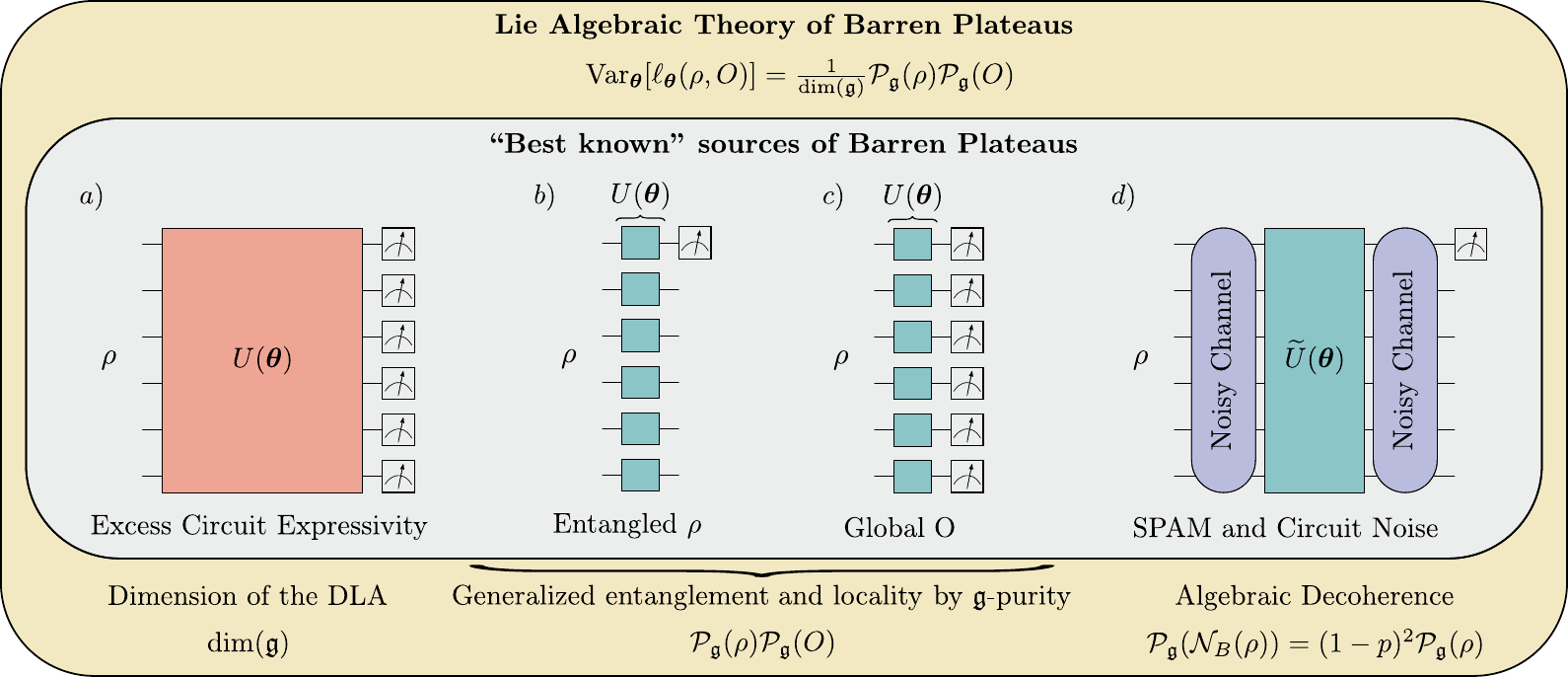}
    \caption{\textbf{Source of barren plateaus under our framework.}
    The ``best known'' sources of BPs include: 
    (a) circuits that are too expressive with $\dim(\g)\sim 4^n$, for any $O$, $\rho$;
    (b) highly entangled initial states $\rho$, even for shallow circuits and local $O$;
    (c) global operators $O$, even for shallow circuits;
    (d) noise channels.
    Our framework unifies these sources in one holistic picture, characterizing the circuit expressivity as $\dim(\g)$,
    the initial state and operator properties by generalized entanglement and locality via $\g$-purity, and noise channels by algebraic decoherence.
     }
    \label{fig:methods_fig}
\end{figure*} 

\subsection{Sources of barren plateaus}

(i) Ansatz expressiveness and the DLA.  First, we observe from
Theorem~\ref{thm:variance-reductive} that the variance is inversely proportional to $\dim(\g)$, which  
directly quantifies the expressiveness of the circuit; that is,
more expressive circuits (with larger $\dim(\g)$) lead to more concentrated loss functions (see Fig~\ref{fig:methods_fig}(a)). In particular, as shown in Corollary~\ref{cor:concentration}, if the dimension of the DLA is in $\Omega(b^n)$ for some $b>2$, then the loss is exponentially concentrated regardless of the initial state or measurement operator.  
Thus, deep circuits with exponential DLAs are always untrainable. On the other hand, if $ \dim(\g)\in\OC(\poly(n))$ (see examples of polynomially scaling DLAs in Refs.~\cite{kokcu2022fixed,schatzki2022theoretical,kazi2023universality,wiersema2023classification}), the expressiveness of the circuit will not induce BPs by itself. As we will see below, this does not preclude the possibility that the initial state or the measurement operator could still lead to exponential concentration.  Importantly, this result proves
(under the assumption that $\rho$ or $O$ is in $i\g$)
the conjecture proposed in Ref.~\cite{larocca2021diagnosing}, where the authors argue precisely that the loss variance should be inversely proportional to the expressive power of the circuit, captured by the DLA dimension.

(ii) Initial quantum state $\rho$. 
Second, we consider the $\g$-purity of a quantum state, $\PC_{\g}(\rho)=\Tr [\rho_{\g}^2]$.
This quantity has been studied in the literature and is known to be a measure of generalized entanglement~\cite{barnum2003generalizations,barnum2004subsystem}. As discussed in Ref.~\cite{barnum2004subsystem}, one can define a generalization of entanglement relative to a
subspace of observables rather than the usual approach,
where distinguished subsystems are used to define the entanglement. 
In standard entanglement theory, the preferred subspaces of observables are given by the local operators that act only on each subsystem, 
and a pure entangled state looks mixed when partially tracing over
one of the subsystems (producing a reduced state). This notion can be generalized by imposing that reduced states of a quantum system only
provide the expectations of some set of distinguished observables, and we say that a state is generalized unentangled relative to the distinguished observables if its reduced state
is pure (cf.\ Fig.~\ref{fig:methods_fig}(b)).
Taking the subspace of observables to be the DLA leads precisely to this generalized measure of entanglement. In this framework, the purity $\PC_{\g}(\rho)$ is maximized ($\rho$ is not generalized entangled) if $\rho$ belongs to the orbit of the highest weight state of $\g$ (see Methods for a definition of highest weight state, as well as a particular case where we specialize to $\rho$ or $O$ simultaneously diagonalizable with a Cartan subalgebra $\mf{h}\subseteq\g$). Essentially, a smaller $\g$-purity implies larger generalized entanglement and, in turn, smaller variances. Hence, if $\PC_{\g}(\rho)\in\OC(1/b^n)$ (high  generalized entangled state), the loss concentrates exponentially regardless of the circuit expressiveness, that is, even if $\dim(\g)$ is polynomial.

(iii) Measurement operator $O$. Finally, we turn to the $\g$-purity of the measurement operator $O$ (cf.\ Fig.~\ref{fig:methods_fig}(c)).
We carry over the notion of generalized entanglement to define a generalized notion of locality: We call an operator $O$ generalized-local
if it belongs to the preferred subspace of observables given by $\g$. In this case, $O_{\g}=O$.  On the other hand, we will call it (fully) generalized nonlocal if $O_{\g}=0$. This definition allows us to establish a hierarchy of locality when $O$ belongs to the DLA's associative algebra (the matrix-product closure of $\g$). Here, one can classify the amount of generalized locality by expressing $O$ as a polynomial of elements of $\g$, and filter them according to the polynomial's degree; 
more local operators will be those that can be expressed as polynomials of lower degree. 
With the above definition, one can readily see that generalized local operators (i.e., those satisfying $O_{\g}=O$)  maximize the variance. 
On the other hand, when $O_{\g}\in\OC(1/b^n)$ (highly generalized nonlocal measurements), the loss exhibits a BP regardless of the DLA dimension.

Putting points (i)-(iii) together, our results in Theorem~\ref{thm:variance-reductive} imply that the loss concentration is completely determined by the expressiveness of the parametrized quantum circuit, the generalized entanglement of the input state, and the generalized locality of the measurement operator.
When the DLA $\g$ 
is not simple, in order for the loss to be exponentially concentrated and a BP to occur, 
each of the components $\g_j$ needs exponentially large DLAs, highly generalized-entangled states, or very generalized-nonlocal measurements. When this condition is not satisfied on some of the components $\g_j$, one can train the loss function only on the signal from these components, but only for those particular components $\g_j$; for the remainder, the variances are exponentially suppressed.

\subsection{Incorporating the effects of noise}

To finish, we consider how SPAM noise and coherent errors affect the variance scaling
(cf.\ Fig.~\ref{fig:methods_fig}(d)). We begin by considering the case of state preparation noise, assuming $O\in i\g$. The effect of noise is to change the initial state from $\rho$ to $\NC_B(\rho)$. From Theorem~\ref{thm:variance-reductive}, we find that if the $\g$-purity of the state decreases (in some or all components $\g_j$ of $\g$), then so will the variance.  This will be the case for global depolarizing noise $\NC_B(\rho)=(1-p)\rho+p\id/2^n$, which takes $\PC_{\g_j}(\rho)$ to $\PC_{\g_j}(\NC_B(\rho))=(1-p)^2\PC_{\g_j}(\rho)$. Moreover, we can see that even local unitaries acting on a single qubit during state preparation (which would preserve the standard purity $\Tr[\rho^2]$ of the state and the standard entanglement) can lead to a generalized-entanglement-induced loss concentration if they rotate the state outside of $i\g$. These examples showcase a form of algebraic decoherence, whereby the state gets entangled with either an actual environment (e.g., the previous case of  depolarizing noise) or with an effective ``algebraic environment'' within the system. 

Next, we consider the presence of measurement errors. We can model this with a channel $\NC_A^{-1}$ acting on $O$, where $\NC_A^{-1}$ denotes the inverse of the noise channel $\NC_A$, which is generally not a physical quantum channel. In this case, if the  $\g$-purity of the measurement operator decreases, i.e., if $O$ loses generalized locality, then the loss concentrates.

Finally, we study the effect of coherent errors that occur during circuit execution. In this case, we can model their impact as uncontrolled unitaries acting along the circuit, which will change the set of generators to $\widetilde{\GC}=\GC\cup \{iK_l\}$, and thus the DLA becomes $\widetilde{\g}\supseteq\g $. Interestingly, this means that coherent noise can increase the expressiveness of the circuit at the cost of potentially decreasing the variance. To know exactly how the variance changes, one needs to study the reductive decomposition of $\widetilde{\g}$, as well as the $\widetilde{\g}$-purities of $\rho$ and $O$ in this new noise-induced DLA.

\section{Discussion}

Finding ways to avoid or mitigate barren plateaus (BPs) has been one of the central topics of research in variational quantum computing. This has led the community to develop a series of good practice guidelines such as: ``global observables where one measures all qubits are untrainable'' or ``too much entanglement leads to BPs.'' While these are widely regarded as being universally true, they are in fact obtained by extrapolating results derived for a specific circuit architecture, and assuming that they will hold in another. Many of these misconceptions are propagated due to a lack of a single unifying theory of BPs that truly connects the results obtained for different architectures. In this work, we present one such theory based on the Lie algebraic properties of the parametrized quantum circuit, the initial state, and the measurement operator. 

Despite the relative simplicity of our main theorem, its implications are extremely important to the field of variational quantum computing. First, it is worth highlighting that unlike other results, our theorems provide exact variance calculations rather than upper or lower bounds, thus allowing us to precisely determine whether we will or will not have a BP. Second, we note that our theorem applies to any (noiseless) deep parametrized quantum circuit, 
independent of its structure, but requiring only knowledge of the DLA of the circuit (and provided that $\rho$ or $O$ are in $i\g$). This makes our results applicable in essentially all areas of variational quantum computing. Conceptually, our results uncover the precise role that entanglement and locality play in the nature of BPs. However, they do so by being defined in terms of the operators in the DLA rather than in terms of local operators determined by the qubit-subsystem decomposition. As such, it is entirely possible for a circuit to be trained on highly entangled initial states using highly nonlocal measurements (i.e., acting on all qubits) as long as they are well aligned with the underlying DLA of the circuit. This goes against one of the
best practices in the field and demonstrates that a more subtle approach is valuable.
It is also worth noting that our algebraic formalism allows us to consider simple noise models such as SPAM errors and coherent noise.  Our analysis shows that noise-induced concentration can be understood as building generalized entanglement in the initial state, reducing the generalized locality of the measurement observable, or increasing the expressive power of the circuit.

Looking ahead, we see different ways to extend our results. For example, our theorems are derived for the case where $\rho$ or $O$ is in $i\g$. Although this case encompasses most algorithms in the literature~\cite{cerezo2020variationalreview,bharti2021noisy,endo2021hybrid}, it could be interesting to generalize our results to operators and measurements not in the DLA. We refer the readers to Ref.~\cite{diaz2023showcasing}, where  exact loss function variance calculations are presented for a parametrized matchgate circuit, which are valid for arbitrary measurements and initial states. We hope that the results therein will serve as a blueprint for a general theory for generic circuits.   Moreover, one can also envision considering more realistic noise settings where noise channels are interleaved with the unitaries.  Clearly, since a noisy parametrized quantum circuit no longer forms a group, our Lie algebraic formalism no longer applies. Similarly, we encourage the community to develop tools to compute the exact variance for circuits that do not form approximate $2$-designs, as these cases are not covered by our main theorems. 
In this case, a Lie algebraic treatment such as the one presented here will not be available. Here, one will have to instead make additional assumptions such as local gates being sampled from a local group, thus mapping the variance evaluation to a Markov chain-like process which can be analyzed analytically~\cite{cerezo2020cost,uvarov2020barren,pesah2020absence,heyraud2023efficient,liu2021presence,barthel2023absence,miao2023isometric,garcia2023barren,letcher2023tight,arrazola2022universal,zhang2023absence}, numerically via Monte Carlo~\cite{napp2022quantifying} and tensor networks~\cite{braccia2024computing,hu2024demonstration}, or studied via XZ-calculus~\cite{zhao2021analyzing}. Moreover, we note that recent results have tied the DLA and the presence or absence of barren plateaus to other phenomena like the simulability of the quantum model~\cite{diaz2023showcasing,cerezo2023does}, to the presence of local minima in the optimization landscape~\cite{larocca2021theory}, and to the reachability of solutions (through the disconnectedness in the geometric manifold of unitaries~\cite{marvian2022restrictions}). We expect that in the near future, the generalized notions of entanglement and locality discussed here will be connected to the resources needed to make a model more or less universal, more or less simulable, and more or less trainable. Such connections will shed additional light on the central role that the DLA plays in characterizing variational quantum computing models.

Note added. A few days before our work was submitted to the arXiv, the manuscript~\cite{fontana2023theadjoint} was posted as a preprint. In Ref.~\cite{fontana2023theadjoint}, the authors provide an independent proof of the conjecture in~\cite{larocca2021diagnosing} and present results similar to those in our theorems and examples.  However, we note that in Ref.~\cite{fontana2023theadjoint} the authors study partial derivative concentration rather than loss concentration as we do. Moreover, our work provides 
a novel conceptual understanding of the variance $\g$-purity terms as forms of generalized entanglement and locality. As such, our work complements that of~\cite{fontana2023theadjoint}, and we encourage the reader to review both manuscripts. 

\section{Methods}
In this section, we first recall a few basic concepts for the DLA, present examples that showcase how our results can be used to understand several BP results in the literature (see Fig.~\ref{fig:methods_fig}), and derive additional theorems that justify some of the assumptions made in the main text. We also present concrete examples where we can analytically compute the $\g$-purities of the state and the measurement operator.  Furthermore, in the Supplemental Information, we include numerical simulations that illustrate our theoretical findings.

\subsection{Symmetries and the DLA structure}

In Eq.~\eqref{eq:reductive} we have argued that the DLA can be generically decomposed as a direct sum of commuting ideals. Here we provide additional intuition as to why this is the case. First, let us begin by denoting as ${\rm comm}(\GC)$ the commutant of the set of generators $\GC$ (i.e., the symmetries of the gate generators),  that is, 
\begin{equation}
    {\rm comm}(\GC)=\{A\in\mathbb{C}^{2^n\times 2^n}\;\;|\;\;[A,H]=0\;\forall H\in\GC\}\,,
\end{equation}
and as $\mathfrak{z}$ the center of ${\rm comm}(\GC)$, i.e.,
\begin{equation}
    \mathfrak{z}=\{B\in{\rm comm}(\GC)\;\;|\;\;[B,L]=0\;\forall L\in{\rm comm}(\GC)\}\,.
\end{equation}
The importance of the commutant ${\rm comm}(\GC)$ and its center $\mathfrak{z}$ span from the fact that they determine the invariant subspaces of the Hilbert space under the action of $G=e^{\mathfrak{g}}$. Moreover, the dimension of the center is directly related to the DLA's decomposition as we can rewrite the decomposition of $\mathfrak{g}$ into commuting ideals as  
\begin{equation}\label{eq:DLA-irrep-decomp}
    \mathfrak{g} = \bigoplus_{\lambda=1}^{\dim(\mathfrak{z})}\id_{m_{\lambda}}\otimes \mathfrak{g}_\lambda\,,
\end{equation}
with $m_{\lambda}$ denoting the multiplicity of the representation $\mathfrak{g}_\lambda$ of $\mathfrak{g}$.  
We may interpret such a direct sum decomposition as saying ``there exists a basis of our Hilbert space that block-diagonalizes the matrices in $\mathfrak{g}$.'' 

Equation~\eqref{eq:DLA-irrep-decomp} shows that a circuit with symmetries will have a rich and complex DLA structure, while universal, unstructured circuits will likely lead to controllable circuits with $\mathfrak{g}=\mathfrak{su}(2^n)$. In particular, we note that while there have been significant recent efforts on computing DLAs~\cite{zimboras2015symmetry,marvian2022restrictions,marvian2023non,kokcu2022fixed,schatzki2022theoretical,kazi2023universality,wiersema2023classification}, this field is still in its infancy, and we predict that  as our understanding and ability to calculate the DLA increases, the results of the present work  will become more and more important for studying the trainability of structured parametrized quantum circuits. 

\subsection{Examples of barren plateaus and their origins}

In this section, we review some results in the literature that have studied different sources of BPs, and we show how these results can be understood by our framework.

(i) Expressiveness. One of the most studied sources of BPs is the expressiveness of the circuit~\cite{mcclean2018barren,holmes2021connecting,marrero2020entanglement,patti2020entanglement,larocca2021diagnosing,friedrich2023quantum,sharma2020trainability,kieferova2021quantum,pesah2020absence,lee2021towards,martin2022barren,grimsley2022adapt,sack2022avoiding}. As previously noted, by expressiveness we mean the breadth of unitaries that the circuit can generate as the parameters are varied.  

For example, let us consider the case where the unitary $U(\thv)$ forms a $2$-design~\cite{dankert2009exact} over $\mathbb{SU}(2^n)$ (see Fig.~\ref{fig:methods_fig}(a)). This is precisely the case studied in the seminal work of Ref.~\cite{mcclean2018barren}. Here, $\g=\mathfrak{su}(2^n)$, meaning that $\dim(\g)=4^n-1$. From Theorem~\ref{thm:variance-reductive}, we find $\bbE_{\thv}[\ell_{\thv}(\rho,O)] = 0$ (because the Lie algebra has no center), and
\small
\begin{align}
        \Var_{\thv}[\ell_{\thv}(\rho,O)] &= \frac{1}{4^n-1} \Tr [O_{\g}^2] \Tr [\rho_{\g}^2]\label{eq:2-design}\\
        &= \frac{1}{4^n-1} \left(\Tr [O^2]-\frac{\Tr[O]^2}{2^n}\right) \left(\Tr [\rho^2]-\frac{1}{2^n}\right).\nonumber
\end{align}
\normalsize
In the second line above, we have used the fact that, because $i\mathfrak{su}(2^n)$ contains all Pauli matrices except the identity $\id$,
we have $O_\g = O- 2^{-n} \Tr[O] \id$ and similarly for $\rho$ (which has $\Tr[\rho]=1$). 
Since $\Tr[O^2]\leq 2^n$ (by assumption) and noting that $\Tr[\rho^2]-2^{-n}\leq 1$, we find that $\Var_{\thv}[\ell_{\thv}(\rho,O)]\in \OC(\frac{1}{2^n})$. Hence, we can see that, regardless of what $O$ and $\rho$ are, the expressiveness of the circuit will always lead to a BP.

(ii) Measurement locality. A second potential source of BPs  is the measurement operator~\cite{cerezo2020cost,uvarov2020barren,kashif2023impact,khatri2019quantum,uvarov2020variational,leadbeater2021f,cerezo2020variational}. For example, consider the task of sending $\rho=\dya{0}^{\otimes n}$ through a circuit $U(\thv)$ composed of general single-qubit gates. Clearly, this parametrized quantum circuit is very inexpressive as it does not generate entanglement; in fact, we can easily find that $\g=\mathfrak{su}(2)^{\oplus n}$. As shown in~\cite{cerezo2020cost}, if the measurement is global (see Fig.~\ref{fig:methods_fig}(c)), e.g.,   $O=X^{\otimes n}$, one has a BP since $\Var_{\thv}[\ell_{\thv}(\rho,O)]\in \OC(1/2^n)$. However,  measuring an operator in $\g$, such as $O=X_1$ (see Fig.~\ref{fig:methods_fig}(c)), does not lead to BPs. Now Theorem~\ref{thm:variance-reductive} indicates that $\bbE_{\thv}[\ell_{\thv}(\rho,O)] = 0$ and 
\begin{align}
        \Var_{\thv}[\ell_{\thv}(\rho,O)] &= \sum_{j=1}^n \frac{1}{3} \Tr[O_{\g_j}^2] \Tr[\rho_{\g_j}^2] \nonumber \\
        &= \frac{1}{3} \Tr[O_{\g_1}^2] \Tr[\rho_{\g_1}^2] = \frac{1}{3} \label{eq:var-local}\\
        &= \frac{1}{3} \Tr_{\mathbb{C}^2} [O_1^2] \left(\Tr_{\mathbb{C}^2}  [\rho_{1}^2]-\frac{1}{2}\right), \nonumber
\end{align}
where $\Tr_{\mathbb{C}^2}$ denotes the trace over the qubit where $O$ acts.
Indeed, we have $\g=\g_1\oplus\dots\oplus\g_n$ with each $\g_j\cong\mathfrak{su}(2)$. Using the orthonormal basis $\{2^{-n/2} X_j, 2^{-n/2} Y_j, 2^{-n/2} Z_j\}$ of $i\g_j$, we find $O_{\g_1}=X_1$ and $O_{\g_j}=0$ for $2\leq j\leq n$, while $\rho_{\g_j}=2^{-n} Z_j$ for all $j$. This gives $\Tr[O_{\g_1}^2]=2^n$, $\Tr[\rho_{\g_1}^2]=2^{-n}$, yielding the value $1/3$ on the second line of Eq.~\eqref{eq:var-local}. 
The last line in Eq.~\eqref{eq:var-local} provides an alternative interpretation, where the observable $O$ and state $\rho$ are reduced on the first qubit to $O_1=X$ and $\rho_1=\dya{0}$, respectively. Note that the last line in Eq.~\eqref{eq:var-local} is the special case $n=1$ of Eq.~\eqref{eq:2-design}. 
Thus, in this specific case, the notions of generalized entanglement and generalized locality match their standard notions based on the subsystem decomposition of the Hilbert space. As such, if we measure a local operator in qubit $j$,  the variance will depend on the purity of the reduced (input) state over qubit $j$.

(iii) Initial state entanglement. Next,  it has been shown that initial state entanglement can lead to BPs~\cite{mcclean2018barren,cerezo2020cost,thanaslip2021subtleties,shaydulin2021importance,abbas2020power,leone2022practical,holmes2021barren}. For instance, consider again the circuit in Fig.~\ref{fig:methods_fig}(c), where $U(\thv)$ is composed of single qubit gates, and let the measurement operator be local, e.g., $O=X_1$. Although we do not expect expressiveness- or locality-induced BPs, if the initial state $\rho$ satisfies a volume law of entanglement (i.e., if the reduced state on any qubit is exponentially close to being maximally mixed), then $\Tr_{\mathbb{C}^2} [\rho_{1}^2]\in\OC(1/2^n)$ and the loss has a BP according to Eq.~\eqref{eq:var-local}. 

(iv) Noise.  Finally, the presence of hardware noise has been shown to induce loss function concentration~\cite{wang2020noise,franca2020limitations,garcia2023effects}. As an example, consider the circuit in Fig.~\ref{fig:methods_fig}(d), where a global depolarizing noise channel $\NC_B(\rho)=(1-p)\rho+p\, 2^{-n}\id$ acts at the beginning of the circuit (e.g., as part of the SPAM errors). Now, we have $\widetilde{\ell}_{\thv}(\rho, O) = (1-p)\ell_{\thv}(\rho, O)+p\, 2^{-n} \Tr[O]$. As $p$ increases, we can see that $\Var_{\thv}[\widetilde{\ell}_{\thv}(\rho,O)]=(1-p)^2\Var_{\thv}[\ell_{\thv}(\rho,O)]$ approaches $0$ and the loss concentrates around $2^{-n} \Tr[O]$. 

\subsection{Deep parametrized quantum circuits form $2$-designs}

Let us here show that if the circuit is deep enough, then it will form a $2$-design on each of the simple or abelian components of the DLA. In particular, we answer the question: How many layers are needed for $U(\thv)$ to become an approximate $2$-design over $G$?

First, let us express the parametrized quantum circuit by factorizing it into layers. That is, we take $U(\thv) = U_{L}(\thv_L)  U_{L-1}(\thv_{L-1})  \cdots  U_1(\thv_1)$ where each layer $U_{l}(\thv_l)$ is given by $
    U_l(\thv_l) = \prod_{k=1}^K e^{-iH_k \theta_{l,k}}$. Next, 
let us denote the ensemble of unitaries generated by the $L$-layered circuit $U(\thv)$
by $\calE_L$, and its $t$-th moment operator $\calM^{(t)}_{\calE_L}: \gl(\calH^{\otimes t})\to \gl(\calH^{\otimes t})$ is the linear superoperator given by
\begin{equation} \label{def:moment-superop}
    \MC_{\EC_L}^{(t)}(\cdot) = \int_{\EC_L} d\mu(U) \, U^{\otimes t}(\cdot) (U\ad)^{\otimes t}\,. 
\end{equation}
The importance of $\calM_{\calE_L}^{(t)}$ arises from the fact that it allows us to compare how expressive, or how close, a given $U(\thv)$ is to the Haar ensemble over $G$. In particular, denoting as $\calM_{G}^{(t)}$  the moment operator for the Haar ensemble over $G$, one can quantify the expressiveness of the circuit via the norm of the superoperator,
\begin{equation} \label{def:ensemble-distance-A}
        \calA_{\calE_L}^{(t)} = \calM_{\calE_L}^{(t)} - \calM_{G}^{(t)} \,.
\end{equation}
For our purposes, we work with the Schatten $p$-norms where $p=1,\infty$ on the spaces of operators $\gl(\calH^{\otimes t})$. Both are spectral norms: given the singular values $s_i(A)\geq 0$ of $A$, we may express these norms as
\begin{equation} \label{def:schatten 1 and infty norms}
    \norm{A}_1 = \sum_{i} s_i(A), \qquad \norm{A}_\infty = \max_i s_i(A) . 
\end{equation}

This induces corresponding operator norms for superoperators $\calM:\gl(\calH^{\otimes t}) \to \gl(\calH^{\otimes t})$ by the usual
\begin{equation} \label{def:induced Schatten p norm}
    \norm{\calM}_p = \sup_{\substack{A\neq 0\\ A\in \gl(\calH^{\otimes t})}} \frac{\norm{\calM(A)}_p}{\norm{A}_p}.
\end{equation}

We say ${\calE_L}$ forms a $G$ $t$-design if $\calA_{\calE_L}^{(t)} = 0$, and we say ${\calE_L}$ forms an $\epsilon$-approximate $G$ $t$-design if $\norm{\calA^{(t)}_{\calE_L}}_\infty \leq \epsilon$ for some $\epsilon>0$. We prove the following theorem in the Supplemental Information:
\begin{theorem}\label{th:design-layers}
The ensemble of unitaries $\calE_L\subseteq G$ generated by an $L$-layered circuit  $U(\thv)$ will form an $\epsilon$-approximate $G$ $2$-design, when the number of layers $L$ is
\begin{equation}\label{eq:bound2 on L}
    L \geq \frac{\log(1/\epsilon)}{\log\paran{1/\norm{\calA^{(2)}_{\calE_1}}_\infty}} ,
\end{equation} where $\calE_1$ is the ensemble generated by a single layer $U_1(\thv_1)$.
\end{theorem} We actually prove that this is true for any $t$, but to estimate the variance we need only $t=2$.

First, we note that Theorem~\ref{th:design-layers} generalizes the result of~\cite{larocca2021diagnosing}, which was only valid for problems where $\g=\mathfrak{su}(2^n)$ or $\mathfrak{u}(2^n)$,
to the case of an arbitrary DLA. Next, we can explicitly see that if the circuit is deep enough, i.e., if it has enough layers, it will always form an $\epsilon$-approximate $G$ $2$-design, thus justifying the assumptions made in the main text. More concretely, if we know what the expressiveness for one layer is,  Theorem~\ref{th:design-layers} allows us to compute exactly how many layers are needed for $\norm{\calA^{(t)}_{\calE_L}}_\infty$ to be smaller than a given $ \epsilon$. We present in the Supplemental Information an example where we can  explicitly construct $\calA^{(2)}_{\calE_1}$ and obtain its largest eigenvalue.  Finally, we note that the expressiveness of a single layer will depend on several choices in the ansatz such as how the gates are distributed in a single layer, and how the parameters are sampled therein. We leave the study of the single layer expressiveness changes as a function of those factors for future work.

\subsection{Variance for a circuit that does not form a $2$-design over $G$}

In the main text, we have assumed that the circuit is deep enough so that it forms a $2$-design over $G$, and Theorem~\ref{th:design-layers} provides bounds on the necessary number of layers for this to happen. Here, we instead ask the question: How different will the variance be from that in Theorem~\ref{thm:variance-reductive} if the circuit forms approximate, rather than an exact, $2$-design over $G$?

For this purpose, let us consider an $L$-layered circuit such as that in Eq.~\eqref{eq:PQC}. In what follows, we do not assume that the circuit forms an approximate $2$-design over $G$.  Just as before, we denote the ensemble of unitaries generated by the $L$-layered circuit $U(\thv)$
by $\EC_L$, and we define its $t$-th moment superoperator $\calM_{\calE_L}$ by Eq. \ref{def:moment-superop}.
Then, denoting as $\Var_{\calE_L}$ and $\Var_{G}$ the variance obtained by sampling circuits over the distribution $\calE_L$ and over the Haar measure over $G$, respectively, we find that the following theorem holds. 
\begin{theorem}\label{th:variance-layers}
Let $\mathfrak{g}\subseteq \mathfrak{u}(2^n)$ be any dynamical Lie algebra, and suppose either $\rho$ or $O$ are in $i\g$ with $\rho$ a density matrix. Then the difference between the loss function variance of an $L$-layered circuit, with ensemble of unitaries $\calE_L$, and the variance for a circuit that forms a $2$-design over $G$ can be bounded as 
\begin{equation}
\abs{\Var_{\calE_L} [\ell_{\thv}(\rho, O)]-\Var_{G} [\ell_{\thv}(\rho, O)]}\leq 3 \paran{\norm{\calA_{\calE_1}^{(2)}}_\infty}^L \norm{O}_1^2 \,,
\end{equation}
where the usual Schatten $p$-norm and induced Schatten $p$-norm are given by Eq. (\ref{def:schatten 1 and infty norms}) and Eq. (\ref{def:induced Schatten p norm}).
\end{theorem}

Theorem~\ref{th:variance-layers} shows that the variance of an $L$-layered circuit will converge exponentially fast in $L$ to the variance obtained by assuming that $U(\thv)$ forms a $2$-design. The rate of convergence is again determined by the single layers expressiveness  $\norm{\calA_{\calE_1}^{(2)}}_\infty$. As such, we can see that the assumption of $U$ being a $2$-design rapidly becomes an exact statement in the deep circuits regime. The proof of Theorem~\ref{th:variance-layers} is given in the Supplemental Information.

\subsection{Variance as a sum over simple or abelian components}

As noted in the main text, the DLA can always be decomposed into a direct sum as in Eq.~\eqref{eq:reductive}. This implies that since $G$ is compact, it must decompose as $G =G_1\times \dots \times G_{k}$, with $G_j=e^{\mf{g}_j}$. Hence, any parametrized quantum circuit $U(\thv) \in G$ can be expressed as a tuple $U=(U_1,\ldots, U_k)$, where we have omitted the dependence on $\thv$ for simplicity of notation. This realization allows us to show that the loss splits into a sum over the summands $\g_j$ in Eq.~\eqref{eq:reductive}. Explicitly, for the case where $\g=\g_1\oplus\g_2$ with $O\in i\g$ and $U(\thv) \in G$, we have
\begin{equation}
O = O_1 + O_2, \quad U(\thv) = (U_1, U_2)\,, \nonumber
\end{equation}
where $O_j\in\g_j$ and $U_j\in G_j$ for $j=1,2$. Hence, the loss function is
\begin{equation}
    \ell_{\thv}(\rho, O) = \Tr [\rho U(\thv)^\dagger O_1 U(\thv)] +\Tr [\rho U(\thv)^\dagger O_2 U(\thv)]\,.\nonumber
\end{equation}
Since $O_1\in \g_1$ commutes with any element of $G_2$, we obtain
\begin{align}
    \Tr [\rho U(\thv)^\dagger O_1 U(\thv)] &= \Tr [\rho (U_1, U_2)\ad O_1 (U_1, U_2)] \nonumber\\
    &= \Tr[ \rho (U_1, 1)^\dagger O_1 (U_1, 1)]\nonumber \\
    & =: \ell_{1}(\rho, O_1) \,.\nonumber
\end{align} 
A similar result holds for the term $\Tr [\rho U(\thv)^\dagger O_2 U(\thv)]$, and the loss function is expressible as a sum of losses in each simple or abelian component: 
\begin{equation}
\ell_{\thv}(\rho, O) = \ell_{1}(\rho, O_1) + \ell_{2}(\rho, O_2)\,.
\end{equation}
Therefore, the mean value and the variance of $\ell_{\thv}(\rho, O)$ can also be expressed as the sum of individual terms in each simple or abelian component
(additional details of the proof are provided in the Supplemental Information).
This statement is formally captured as follows: 
\begin{proposition} \label{thm:sums-E-and-V}
Let $O$ be in $i\g$ such that $O= \sum_{i=1}^k O_i$ respect the decomposition Eq.~\eqref{eq:reductive} of $\g$ into simple and abelian components.  Then, for any $\rho$, the expectation and variance of the loss function split into sums as
    \begin{equation}\begin{split}
        \bbE_{\thv}[\ell_{\thv}(\rho,O)] &= \sum_{i=1}^k \bbE_{G_i}[\ell_{i}(\rho,O_i)], \\
        \Var_{\thv}[\ell_{\thv}(\rho,O)] &= \sum_{i=1}^k \Var_{G_i}[\ell_{_i}(\rho,O_i)] ,
    \end{split}\end{equation} where $\ell_{i}(\rho,O_i) = \Tr[(1,\dots,U_i,\dots, 1) \rho (1,\dots,U_{i}, \dots, 1)^\dagger O_i]$ 
    corresponds to the loss in the $G_i$-component of $G$.
\end{proposition}
We note that the previous result is equivalently valid if one exchanges the roles of $O$ and $\rho$, that is, if $\rho\in i\g$ and $O$ is any Hermitian observable.

\subsection{Special case: $\rho$ or $O$ simultaneously diagonalizable with a Cartan subalgebra}

Here, we focus on the case where one chooses a Cartan subalgebra $\h\subset \g$ and $\rho$ (or $O$) is simultaneously diagonalizable with $\h$, i.e., $[\rho, \h ]=0$. A basis of $\calH$ diagonalizing $\h$ is called a weight basis, and when $\rho$ is diagonalized by this basis, the $\g$-purity is straightforward to compute. We first demonstrate this situation by revisiting a calculation from Ref.~\cite{larocca2021diagnosing}, which concerns the case where $\g=\mathfrak{su}(2)$ is irreducibly represented in $\HC$. We then generalize this calculation to an arbitrary simple DLA $\g$.

Let $\g$ be the irreducible spin-$S$ representation of $\mathfrak{su}(2)$ acting on $\calH = \C^{2S+1}$.  We have a basis for $i\g$ given by the usual Hermitian spin matrices $\{S_x,S_y,S_z\}\subseteq \gl(\calH)$. One can pick the Cartan subalgebra $i\h = \mathrm{span}_{\mathbb{R}}(S_z)$. A convenient orthonormal basis for $\calH$ is then given by $\{\ket{m}\}$, the eigenvectors of $S_z$, which satisfy
\begin{equation}
    S_z \ket{m} = m \ket{m}\,, \qquad m=-S,-S+1,\ldots, S-1, S\,.
\end{equation}
Let us consider the case where $\rho= \ket{m}\bra{m}$ and $O\in i\g$. To use Theorem~\ref{thm:variance-reductive}, we will need to compute $\rho_\g$. Such projection can be easily computed by performing a change of basis of $\g_\C = \mathrm{span}_\C\g$ from the spin matrices to $\{S_0,S_+,S_-\}$, where
\begin{equation}
    S_0 = \frac{1}{\sqrt{c_z}}S_z, \,\,\, S_+ = \frac{1}{\sqrt{c_+}}(S_x+ iS_y), \,\,\, S_- = \frac{1}{\sqrt{c_+}}(S_x- iS_y),\nonumber
\end{equation} 
with the $c$'s being normalization constants such that $\Tr [S_0^2]= 1$ and $\Tr [S_+^\dagger S_+ ]= \Tr [S_-^\dagger S_-] = 1$ (in particular we find that $c_z = \Tr [S_z^2] = \frac{1}{3} S(S+1)(2S+1)$). On the basis $\{\ket{m}\}$ of $\calH$, $S_z$ is diagonal, $S_+$ is upper triangular nilpotent, and $S_-$ is lower triangular nilpotent. A direct calculation shows that $\Tr[S_+\ad \rho]=\Tr[S_-\ad \rho]=0$, and thus
\begin{equation}\label{eq:reduced-spin}
    \rho_\g = \rho_\h =  \frac{1}{c_z}\Tr[S_z \rho] S_z = \frac{m}{c_z} S_z.
\end{equation}
Then, replacing  Eq.~\eqref{eq:reduced-spin} into Theorem~\ref{thm:variance-reductive} and using $\dim(\g) = 3$ leads to
\begin{equation}\label{eq:var-spin}
    \Var_{\thv}[\ell_{\thv}(\rho,O)] = 
     \frac{m^2\Tr[O^2]}{S(S+1)(2S+1)}\,.
\end{equation} 
Note that $S(S+1)$ is exactly the action of the Casimir element on $\calH$. For the case of $O=S_z/S$ (where we normalized the measurement operator so that $\norm{O}=1$), we find that Eq.~\eqref{eq:var-spin} becomes $\Var_{\thv}[\ell_{\thv}(\rho,O)] =\frac{m^2}{3 S^2}$. Hence, we can see that the loss function concentrates depending on the value $m$, i.e., on the $\g$-purity or the generalized entanglement present in the initial state.

The previous example can be used as a starting point to calculate the loss variance for $\rho$ being a sum of weight states. In particular, consider the case where $\g$ is simple. We can always find an orthonormal basis of weight spaces for $\calH$ by simultaneously diagonalizing a Cartan subalgebra $\h$ of $\g$. This means that we have an orthonormal basis $\{\ket{v}\}$ for $\HC$, so that $\ket{v}$ has weight $\lambda_v$, i.e.
\begin{equation}\label{def:weight-basis}
    H_j \ket{v} = \lambda_v(H_j)\ket{v} , \qquad \text{for all }\; H_j\in i\h .
\end{equation} 
Note that the eigenvalues $\lambda_v(H_j)$ are real as $H_j$ are Hermitian.
Then a pure weight state with respect to $\h$ has the form $\rho = \ket{v}\bra{v}$. We say that $\rho$ is a weight state if
it is diagonal with respect to the weight basis, i.e., if 
\begin{equation}\label{eq:rho-convex}
    \rho = \sum_v r_v \ket{v}\bra{v}, \qquad r_v \geq 0 , \;\; \sum_v r_v = 1\,.
\end{equation}
As shown in the Supplemental Information, given an
orthonormal basis $\{H_j\}_{j=1}^{\dim(\h)}$ of $i\h$, one can find that
\begin{equation}\label{eq:convex-weight}
    \rho_\g = \rho_\h = \sum_{j=1}^{\dim (\h)} \lambda_\rho(H_j) H_j, 
    \quad\text{with}\;\;
    \lambda_\rho = \sum_v r_v \lambda_v.
\end{equation} 
Hence,  
\begin{equation}
    \Tr [\rho_\g^2] = \sum_{j=1}^{\dim (\h)} \lambda_\rho(H_j)^2 = \norm{\lambda_\rho}^2.
\end{equation} 
Then, plugging into Theorem~\ref{thm:variance-reductive} yields the following corollary. 
\begin{corollary} \label{cor:rho-general-weight-state}
    Let  $\g$ be a simple Lie algebra with a Cartan subalgebra $\h$. Let $O\in i\g$, and $\rho$ be a density matrix with $[\rho,\h]=0$, that is, a weight state as in Eq.~\eqref{eq:rho-convex}. Then the variance of the loss function is given by
    \begin{equation}\label{eq:variance-convex}
        \Var_{\thv}[\ell_{\thv}(\rho,O)] = \frac{\Tr[O^2]}{\dim (\g)} \norm{\lambda_\rho}^2 \,.
    \end{equation}
\end{corollary}

Here, we find it important to make three remarks. First, by symmetry, the result of
Corollary~\ref{cor:rho-general-weight-state} also holds when the roles of $\rho$ and $O$ are reversed, i.e., when $\rho\in i\g$ and
$[O,\h]=0$. Second, the variance does not depend on the choice of Cartan subalgebra, only on the norm of the weight $\lambda_\rho$ corresponding to the weight state $\rho$, because all Cartan subalgebras are conjugate to each other. Finally, we note that among all weights $\lambda_v$ in an irreducible representation $\HC$, the norm $\norm{\lambda_v}$ is maximized when $\ket{v}$ is the highest weight vector $\ket{hw}$.
Thus, Eq.~\eqref{eq:variance-convex} implies that among all weight states as in Eq.~\eqref{eq:rho-convex}, the variance is maximized when $\rho$ is a highest weight state:
\begin{align}
    &\argmax_{\rho \text{ as in Eq.~\eqref{eq:rho-convex}}} \Var_{\thv}[\ell_{\thv}(\rho,O)]=\dya{hw}\,.
\end{align}

From here, we can derive an estimate for the variance as follows. Consider the bilinear form on $\g$ given by $(A,B)=\Tr_{\C^N} [AB]$ where the trace is taken over the defining representation $\C^N$ of $\g=\su(N)$, $\mathfrak{so}(N)$ or $\mathfrak{sp}(N)$; this is the standard bilinear form normalized so that $(H,H)=2$ or $4$ for simple coroots $H\in i\h$ (for $\g=\su(N)$, these are the diagonal matrices with $0,\dots,0,1,-1,0,\dots,0$ on the diagonal). The highest weight $\lambda_{hw}$ of every irreducible finite-dimensional representation $\HC$ of $\g$ has the property that $\lambda_{hw}(H)$ is a non-negative integer for any simple coroot $H$. We assume that these values will not grow with the system size $n$. Noting that $\inprod{A,B} = \frac{\dim(\HC)}{N} (A,B)$ for Herimitian operators $A,B\in i\g$, we see that the basis vectors $H_j$ of $i\h$ scale as $\sqrt{N/2^n}$ times the simple coroots of $i\h$. Therefore,
\begin{align}
\norm{\lambda_{hw}}^2 \in \Theta\left( \frac{N}{2^n} \right) = \Theta\left( \frac{\sqrt{\dim(\g)}}{2^n} \right).
\end{align}
Plugging this in Eq.~\eqref{eq:variance-convex}, we obtain the upper bound
\begin{equation}\label{eq:variance-convex2}
 \Var_{\thv}[\ell_{\thv}(\rho,O)] \in \OC\left( \frac{\Tr[O^2]}{2^n \sqrt{\dim(\g)}} \right).
\end{equation} 
Hence, we see that if $\Tr[O^2]\leq 2^n$ and $\dim(\g)\in\Omega(b^n)$ with $b>1$, then $\Var_{\thv}[\ell_{\thv}(\rho,O)] \in \OC\left( \frac{1}{\sqrt{b^n}} \right)$; indicating that if the DLA is simple and of exponential dimension, then the loss will always have a BP regardless of the measurement operator and initial state.

\section*{DATA AVAILABILITY}
Data generated and analyzed during the current study are
available from the corresponding author upon reasonable
request.

\bibliography{quantum}

\section*{Acknowledgments}

We thank Zoe Holmes, Mark Wilde and Akram Touil for fruitful discussions about superoperator norms. We are extremely grateful to Lukasz Cincio for his help with numerical calculations. MR and COM were supported by the Laboratory Directed Research and Development Program and Mathematics for Artificial Reasoning for Scientific Discovery investment at the Pacific Northwest National Laboratory, a multiprogram national laboratory operated by Battelle for the U.S. Department of Energy under Contract DE-AC05- 76RLO1830. 
BNB was supported in part by a Simons Foundation grant No. 584741. 
FS and MC acknowledge support by the Laboratory Directed Research and Development program of Los Alamos National Laboratory (LANL) under project numbers 20230049DR and 20230527ECR. 
AFK was supported in part by the National Science Foundation under award No. 1818914: PFCQC: STAQ: Software-Tailored Architecture for Quantum co-design. ML was supported by  the Center for Nonlinear Studies at LANL.  MC was initially supported by LANL ASC Beyond Moore’s Law project.
\\

\section*{AUTHOR CONTRIBUTIONS}
The project was conceived by BNB, COM, ML and MC. Theoretical results were derived by BNB, MR and MC, with some independently verified by ML and MC. Numerical simulations were performed by FS and MC. The manuscript was mainly written by MC, MR, BNB and AFK. All authors contributed to the manuscript review process.

\section*{COMPETING INTERESTS}
The authors declare no competing interests.

\clearpage
\newpage
\onecolumngrid

\setcounter{lemma}{0}
\setcounter{proposition}{0}
\setcounter{theorem}{0}
\setcounter{figure}{0}

\renewcommand\appendixname{Supplementary Note}
\renewcommand\figurename{Supp. Fig.}

\section*{Supplemental Information for ``A Lie algebraic theory of Barren Plateaus for Deep Parameterized Quantum Circuits''}

This Supplemental Information contains additional details and complete proofs of our main results.

\section{Proof of Proposition 1}

In this section, we give the proof of Proposition 1, which we recall for convenience of the reader:
\setcounter{theorem}{2}

\begin{proposition} \label{thm:sums-E-and-V-SI}
Let $O$ be in $i\g$ such that $O= \sum_{i=1}^k O_i$ respect the decomposition Eq.~(6) of $\g$ into simple and abelian components.  Then, for any $\rho$, the expectation and variance of the loss function split into sums as
    \begin{equation}\begin{split}
        \bbE_{\thv}[\ell_{\thv}(\rho,O)] &= \sum_{i=1}^k \bbE_{G_i}[\ell_{i}(\rho,O_i)], \\
        \Var_{\thv}[\ell_{\thv}(\rho,O)] &= \sum_{i=1}^k \Var_{G_i}[\ell_{_i}(\rho,O_i)] ,
    \end{split}\end{equation} where $\ell_{i}(\rho,O_i) = \Tr[(1,\dots,U_i,\dots, 1) \rho (1,\dots,U_{i}, \dots, 1)^\dagger O_i]$ 
    corresponds to the loss in the $G_i$-component of $G$.
\end{proposition}

While in the main text we have omitted the parameter dependence of $U_i$ and $\ell_i$ for ease of notation, in the proof below we will make it explicit again. That is, we will denote $U_i:=U_i({\thv_i})$ and $\ell_i:=\ell_{\thv_i}$.

\begin{proof}\let\qed\relax
For simplicity of notation, it suffices to demonstrate the case $k=2$, as the general case is analogous. Since $O\in i\g$ and $U(\thv) \in G$, we have
\[
    O = O_1 + O_2, \quad U(\thv) = (U_1({\thv_1}), U_2({\thv_2})) .
\] 
Here $\thv_i$ $(i=1,2)$ are coordinates on the Lie group $G_i=e^{\g_i}$, and $\thv=(\thv_1,\thv_2)$ are coordinates on the Cartesian product $G = G_1\times G_2$.
The loss function is then
\begin{equation}
    \ell_{\thv}(\rho, O) = \Tr[\rho U(\thv)^\dagger O_1 U(\thv)] +\Tr[\rho U(\thv)^\dagger O_2 U(\thv)] .
\end{equation} 
Notice that since $O_1\in \g_1$, it commutes with every element of $G_2$, and so
\begin{align*}
    \Tr[ \rho U(\thv)^\dagger O_1 U(\thv)] &= \Tr[ \rho \,(U_1({\thv_1}), U_2({\thv_2}))^\dagger O_1 (U_1({\thv_1}), U_2({\thv_2})) ] \\
    &= \Tr[ \rho \,(U_1({\thv_1})^\dagger, U_2({\thv_2})^\dagger) (1, U_2({\thv_2})) O_1 (U_1({\thv_1}), 1)] \\
    &= \Tr[ \rho \,(U_1({\thv_1}), 1)^\dagger O_1 (U_1({\thv_1}), 1)] \\
    & =: \ell_{\thv_1}(\rho, O_1) .
\end{align*} 
Hence, the loss function is expressed as a sum of independent random variables depending on $\thv_1,\thv_2$:
\begin{equation}
\ell_{\thv}(\rho, O) = \ell_{\thv_1}(\rho, O_1) + \ell_{\thv_2}(\rho, O_2) .
\end{equation} 
Since the expectation (respectively, variance) of a sum of independent random variables is equal to the sum of expectations (respectively, variances), we obtain:
\begin{equation}\begin{split}
    \bbE_{\thv} [\ell_{\thv}(\rho, O)] &= \bbE_{\thv} [\ell_{\thv_1}(\rho, O_1)] + \bbE_{\thv} [\ell_{\thv_2}(\rho, O_2)] ,\\ 
    \text{Var}_{\thv} [\ell_{\thv}(\rho, O)) &= \text{Var}_{\thv} [\ell_{\thv_1}(\rho, O_1)] + \text{Var}_{\thv} [\ell_{\thv_2}(\rho, O_2)].
\end{split}\end{equation} 
But since we are working with products of the normalized Haar measure and $\ell_{\thv_i}(\rho,O_i)$ is constant with respect to $\thv_j$ for $j\neq i$, 
we have $\bbE_{\thv}[\ell_{\thv_i}(\rho, O)] = \bbE_{\thv_i}[\ell_{\thv_i}(\rho, O_i)]$, and likewise for the variance.

\end{proof}

\section{Proof of Theorem 1 and of Eq.~(28)}
\label{sec:theorem_1_proof}

Let us now present the proof of Theorem 1, which we  restate for convenience. 
\setcounter{theorem}{0}

\begin{theorem}\label{thm:variance-reductive-SI}
    Suppose that $O\in i\g$ or $\rho \in i\g$, where the DLA $\g$ is as in Eq.~(6).
    Then the mean of the loss function vanishes for the semisimple component $\g_1 \oplus \dots\oplus \g_{k-1}$ and leaves only abelian     contributions:   
    \begin{equation}\label{eq:exq-SI}
       \bbE_{\thv}[\ell_{\thv}(\rho,O)] = \Tr [\rho_{\g_k} O_{\g_k}] \,.
    \end{equation}
Conversely, the variance of the loss function vanishes for the center $\g_k$ and leaves only simple contributions:
\begin{align}\label{eq:var-SI}
    \Var_{\thv} [\ell_{\thv} (\rho,O)] = \sum_{j=1}^{k-1} \frac{\PC_{\g_j}(\rho) \PC_{\g_j}(O)}{\dim (\g_j)} \,.
\end{align}       
\end{theorem}

\begin{proof}\let\qed\relax
By symmetry, we can assume that $O\in i\g$. Due to Proposition 1 of the main text, both the expectation and variance split as sums over the components of $\g$. Therefore, it is enough to consider the cases where $\g$ is simple or abelian. The abelian case follows from the next lemma, which is obvious but is included here for completeness.

\begin{lemma} \label{lem:variance-g-abelian2}
Let either $[O,\g] =0$ or $[\rho, \g]=0$. Then the loss function $\ell_{\thv}(\rho,O)$ is constant. In particular, its mean and variance are given by: 
\begin{equation}
\begin{split}
    \bbE_{\thv}[\ell_{\thv}(\rho,O)] &= \Tr [\rho O] \,, \\
    \Var_{\thv}[\ell_{\thv}(\rho,O)] &= 0 \,.
\end{split}
\end{equation}
\end{lemma}
\begin{proof}\let\qed\relax
For concreteness, suppose that $[O,\g] =0$. Then $O$ commutes with $G$. Hence, for any $U(\thv) \in G$,
    \[
        \ell_{\thv}(\rho,O) = \Tr[U(\thv) \rho U(\thv)^{\dagger} O] = \Tr[\rho U(\thv)^{\dagger} O U(\thv)] = \Tr[\rho O].
    \] 
    This function is constant with respect to $\thv$ and the result follows.
\end{proof}

From now on, let us restrict to the case where $\g$ is simple. By definition, the expectation of the loss function is 
\begin{equation}\label{eqn:expectation}
    \bbE_{\thv} [\ell_{\thv}(\rho, O)] = \int_G d\mu(U(\thv)) \, \Tr [U(\thv) \rho U(\thv)^\dagger O] ,
\end{equation}
where the integral is taken with respect to the Haar measure on $G$. Using the linearity of the trace, we can rewrite this as
\begin{equation} \label{eqn:expectation via first moment operator}
    \bbE_{\thv} [\ell_{\thv}(\rho, O)] = \Tr [\calM_G^{(1)}(\rho) O]
    =\inprod{\calM_G^{(1)}(\rho), O},
\end{equation}
where the first moment operator $\calM_G^{(1)}$ is defined by (cf.\ Eq.~(19) of the main text):
\begin{equation} \label{eqn:first moment operator}
\calM_G^{(1)}(\rho) = \int_G d\mu(U(\thv)) \, U(\thv) \rho U(\thv)^\dagger .
\end{equation}
On the right side of Eq.~\eqref{eqn:expectation via first moment operator}, we have expressed the trace in terms of
the Hilbert--Schmidt inner product $\inprod{A,B} = \Tr [A^\dagger B]$, using that $\calM_G^{(1)}(\rho)$ is Hermitian.
If we use the cyclicity of trace in Eq.~\eqref{eqn:expectation} to replace $\Tr [U(\thv) \rho U(\thv)^\dagger O]$ with $\Tr[\rho U(\thv)^{\dagger} O U(\thv)]$ in Eq.~\eqref{eqn:expectation} and the invariance of the Haar measure under $U\mapsto U\ad$, we can also derive
\begin{equation} \label{eqn:expectation via first moment operator2}
    \bbE_{\thv} [\ell_{\thv}(\rho, O)] = \Tr [\rho \calM_G^{(1)}(O)]
    =\inprod{\rho, \calM_G^{(1)}(O)}.
\end{equation}

More generally, let us recall the $t$-moment operator (cf.\ Eq.~(19)):
\begin{equation} \label{eqn:t moment operator}
\calM_G^{(t)}(A) = \int_G d\mu(U(\thv)) \, U(\thv)^{\otimes t} A (U(\thv)^\dagger)^{\otimes t} .
\end{equation}
Here $A$ is any linear operator acting on the $t$-th tensor power of the Hilbert space $\HC^{\otimes t}$.
Below, we will apply it in the particular case $t=2$ and $A=O \otimes O$.
A similar calculation as above, using the linearity and cyclicity of the trace, shows that the linear operator $\calM_G^{(t)}$ is self-adjoint with respect to the Hilbert--Schmidt inner product on $\gl(\HC^{\otimes t})$:
\begin{equation}\label{eqn:t moment self-adjoint}
\begin{split}
\inprod{\calM_G^{(t)}(A), B} 
&= \int_G d\mu(U(\thv)) \, \Tr [ U(\thv)^{\otimes t} A^\dagger (U(\thv)^\dagger)^{\otimes t} B] \\
&= \int_G d\mu(U(\thv)) \, \Tr [ A^\dagger (U(\thv)^\dagger)^{\otimes t} B U(\thv)^{\otimes t} ] \\
&= \inprod{A,\calM_G^{(t)}(B)} .
\end{split}
\end{equation}

The next lemma, which follows from the left-invariance of the Haar measure, is the key ingredient in the proof of the theorem.
The result of the lemma is well known and forms the foundation for the Weingarten Calculus (see Ref.~\cite{collins2022weingarten}
and the references therein).

\begin{lemma} \label{lem:moments invariance}
The $t$-moment operator is an orthogonal projection from the space of linear operators $\gl(\HC^{\otimes t})$ onto its subspace $\gl(\HC^{\otimes t})^G$ of $G$-invariants,
where $G$ acts on $\gl(\HC^{\otimes t})$ by conjugation: $A \mapsto g^{\otimes t} A (g^\dagger)^{\otimes t}$ for $g\in G$. 
\end{lemma}
\begin{proof}\let\qed\relax
By the left-invariance of the Haar measure, we have:
\begin{equation} \label{eqn:t moment operator2}
g^{\otimes t} \calM_G^{(t)}(A) (g^\dagger)^{\otimes t} = \int_G d\mu(U(\thv)) \, (gU(\thv))^{\otimes t} A (g(U(\thv))^\dagger)^{\otimes t} 
= \calM_G^{(t)}(A).
\end{equation}
Hence, $\calM_G^{(t)}(A)$ is $G$-invariant, i.e., fixed by every $g\in G$.
Conversely, if $A$ is $G$-invariant, then it commutes with the action of $G$, and we have
$\calM_G^{(t)}(A) = A$ as in the proof of Lemma~\ref{lem:variance-g-abelian2}. Therefore, $\calM_G^{(t)}$ is a projection. 
Since $\calM_G^{(t)}$ is self-adjoint, it is an orthogonal projection.
\end{proof}

The Haar measure on a compact Lie group $G$ is both left and right invariant. Similarly to Eq.~\eqref{eqn:t moment operator2}, we also have
\begin{equation} \label{eqn:t moment operator2r}
\calM_G^{(t)} \left( g^{\otimes t} A (g^\dagger)^{\otimes t} \right) = \int_G d\mu(U(\thv)) \, (U(\thv)g)^{\otimes t} A ((U(\thv)g)^\dagger)^{\otimes t} 
= \calM_G^{(t)}(A).
\end{equation}
This means that if we view $\calM_G^{(t)}$ as a linear transformation
\begin{equation} \label{eqn:t moment operator3}
\calM_G^{(t)} \colon \gl(\HC^{\otimes t}) \to \gl(\HC^{\otimes t})^G ,
\end{equation}
then it is a homomorphism of representations of $G$, i.e., it commutes with the $G$-actions on both sides of the equation.
Moreover, due to the unitarity and complete reducibility of the representation of $G$ on $\gl(\HC^{\otimes t})$,
the latter splits as a direct sum of $\gl(\HC^{\otimes t})^G$ and its orthogonal complement.

Now let us get back to the proof of the theorem in the case when $\g$ is simple and $O\in i\g$.
Our construction of the DLA $\g$ as a subalgebra of $\gl(\HC)$ is justified by the fact that any nontrivial representation of a simple Lie algebra is faithful. 
Moreover, we are assuming that the Lie group $G=e^\g$ has a unitary representation on $\HC$, which means that $\g\subseteq\u(\HC)$.
As $\g$ is a subalgebra of $\gl(\HC)$, it is also a subrepresentation for the adjoint action of $G$.
Then the restriction of $\calM_G^{(1)}$ to $\g$ gives a homomorphism
\begin{equation} \label{eqn:t moment operator4}
\calM_G^{(1)} \colon \g \to \g^G .
\end{equation}
But $\g^G$ is the center of $\g$, which is trivial. Therefore, $\calM_G^{(1)}(O)=0$, proving that the expectation of the
loss function is zero.

Now we find the variance of the loss function. 
Using that $\Tr[A \otimes B] = \Tr[A] \Tr[B]$, we have:
\begin{equation} \label{eqn:variance2}
\begin{split}
\Var_{\thv} [\ell_{\thv}(\rho, O)] &= \bbE_{\thv}  [\ell_{\thv}(\rho, O)^2] \\
&= \int_{G} d\mu(U(\thv)) \, \paran{\Tr[ U(\thv) \rho U(\thv)^\dagger O] }^2  \\
    &= \int_{G} d\mu(U(\thv)) \, \Tr[ U(\thv)^{\otimes 2} \rho^{\otimes 2} (U(\thv)^\dagger)^{\otimes 2} O^{\otimes 2} ] \\
    &= \Tr [\calM^{(2)}(\rho^{\otimes 2}) O^{\otimes 2}] \,.
\end{split}
\end{equation}
As above, we can express the last equation in terms of the Hilbert--Schmidt inner product:
\begin{equation} \label{eqn:variance3}
\Var_{\thv} [\ell_{\thv}(\rho, O)] = \inprod{\calM^{(2)}(\rho^{\otimes 2}), O^{\otimes 2}}
= \inprod{\rho^{\otimes 2}, \calM^{(2)}(O^{\otimes 2})}.
\end{equation}

Recall that $O\in i\g$, so $O^{\otimes 2} \in\g\otimes\g$.
Since $\g\subset\gl(\HC)$, we have $\g\otimes\g\subset\gl(\HC)\otimes\gl(\HC) \cong \gl(\HC^{\otimes 2})$ is a subrepresentation for the adjoint action of $G$.
As above, the restriction of the homomorphism $\calM_G^{(2)}$ to $\g^{\otimes 2}$ gives a homomorphism
\begin{equation} \label{eqn:t moment operator5}
\calM_G^{(2)} \colon \g\otimes\g \to (\g\otimes\g)^G ,
\end{equation}
which is an orthogonal projection.
The following lemma is well known (see, e.g., \cite{kirillov2008introduction}, Proposition 6.54 and the discussion after it).

\begin{lemma} \label{lem:casimir}
Let $\g$ be a simple subalgebra of $\u(\HC) \cong \u(N)$, where $\HC\cong\mathbb{C}^N$ is an $N$-dimensional Hilbert space. Then the space of $G$-invariants $(\g\otimes\g)^G$ is $1$-dimensional and spanned by the split quadratic Casimir element
\begin{equation} \label{eqn:casimir}
C = \sum_{j=1}^{\dim(\g)} B_j \otimes B_j ,
\end{equation}
where $\{B_j\}_{j=1}^{\dim(\g)}$ is an orthonormal basis for $i\g$ with respect to the Hilbert--Schmidt inner product.
\end{lemma}
\begin{proof}\let\qed\relax
As $\g\subseteq\u(N)$ and is simple, it is a compact real Lie algebra (i.e., the Lie group $G\subseteq\mbb{U}(N)$ is compact); hence, its complexification $\g_\C$ is a simple Lie algebra whenever $\g$ is simple (see e.g.\ \cite{knapp2013lie}, Chapters IV and VI). 

It is well known that for any simple Lie algebra $\g_\C$ over $\C$, the space of $G$-invariants $(\g_\C\otimes\g_\C)^G = \C C$, where $C$ is given by Eq.~\eqref{eqn:casimir} and
$\{B_j\}_{j=1}^{\dim(\g)}$ is an orthonormal basis for $\g_\C$ with respect to some nondegenerate symmetric $G$-invariant bilinear form $(\cdot,\cdot)$. 
Moreover, $C$ does not depend on the choice of basis for $\g_{\mathbb{C}}$.
These claims follow from Schur's Lemma, because the adjoint action of $\g_\C$ on itself is irreducible
and after we identify $\g_\C\otimes\g_\C \cong \gl(\g_\C)$ as $G$-representations, the $G$-invariant elements correspond to
homomorphisms $\g_\C\to\g_\C$, which are only scalar multiples of the identity.

We take $(A,B)=\Tr[AB]$ to be the trace form (and again by Schur's Lemma, any two such forms are scalar multiples of each other). Now the proof follows from the observation that $\inprod{A,B} = \Tr[AB]$ for all $A,B\in i\g$. 
\end{proof}

Using that $(\g\otimes\g)^G = \mathbb{R} C$, we can compute $\calM^{(2)}(O^{\otimes 2})$ by taking the orthogonal projection of $O^{\otimes 2} \in \g\otimes\g$ onto $C$:
\begin{equation} \label{eqn:second moment operator is 1d orth proj}
    \calM^{(2)}(O^{\otimes 2}) = \frac{\inprod{C,O^{\otimes 2}}}{\inprod{C,C}} C.
\end{equation}
Combining this with Eq.~\eqref{eqn:variance3}, we get:
\begin{equation} \label{eqn:variance4}
\Var_{\thv} [\ell_{\thv}(\rho, O)] 
= \frac{\inprod{\rho^{\otimes 2}, C} \inprod{C,O^{\otimes 2}}}{\inprod{C,C}}
= \frac{\inprod{C,\rho^{\otimes 2}} \inprod{C,O^{\otimes 2}}}{\inprod{C,C}} \,,
\end{equation}
where for the last equality we used that all elements of $\g\otimes\g$ are Hermitian.

Now we compute each of the inner products in the above formula. For the denominator, we have
\begin{equation}
\inprod{C,C} 
= \sum_{i,j=1}^{\dim(\g)} \inprod{B_i \otimes B_i, B_j \otimes B_j}
= \sum_{i,j=1}^{\dim(\g)} \inprod{B_i, B_j} \inprod{B_i,B_j} = \dim(\g) .
\end{equation} 
For the other two terms, we apply the next lemma, which follows from simple linear algebra.

\begin{lemma} \label{lem:moments invariance-2}
For any Hermitian operator $H\in i\u(\HC)$, we have that
\begin{equation}
\inprod{C,H^{\otimes 2}} = \PC_{\g}(H)
\end{equation} 
is the $\mf{g}$-purity of $H$ as defined in Eq.~(7) of the main text.
\end{lemma}
\begin{proof}\let\qed\relax
The orthogonal projection of $H$ into $\g_\C$ is given by
\begin{equation}\label{eq:Hg}
H_\g
= \sum_{j=1}^{\dim(\g)} \inprod{B_j, H} B_j ,
\end{equation}
where $\{B_j\}_{j=1}^{\dim(\g)}$ is any orthonormal basis for $\g_\C$ with respect to the Hilbert--Schmidt inner product $\inprod{\cdot,\cdot}$.
Then
\begin{equation}\label{eq:Hg2}
\inprod{H_\g,H_\g}
= \sum_{j=1}^{\dim(\g)} \abs{\inprod{B_j, H}}^2 .
\end{equation}
Since the restriction of $\inprod{\cdot,\cdot}$ to the space $i\u(\HC)$ of Hermitian operators is positive definite, we can pick the basis vectors $B_j$ to be all in $i\g$. This gives that $H_\g\in i\g$. Hence, $\inprod{H_\g,H_\g} = \Tr[H_\g^2]$, which proves the equality in Eq.~(7).

On the other hand, when all $B_j\in i\g$, we have
\begin{align*}
\inprod{C,H^{\otimes 2}} 
= \sum_{j=1}^{\dim(\g)} \inprod{B_j \otimes B_j, H^{\otimes 2}}
= \sum_{j=1}^{\dim(\g)} \inprod{B_j, H}^2 
= \inprod{H_\g,H_\g},
\end{align*}
because in this case $\inprod{B_j, H} \in\mathbb R$.
\end{proof}
Plugging the result of the above lemma in Eq.~\eqref{eqn:variance4} completes the proof of Theorem 1.
\end{proof}

We end this section with a proof of Eq.~(28). 
Let us pick an orthonormal basis $\{B_j\}_{j=1}^{\dim(\g)}$ for $i\g$, so that
it contains as a subset the basis $\{H_j\}_{j=1}^{\dim(\h)}$ for the Cartan subalgebra $i\h$
and the $B_j$ that are not in $i\h$ are complex linear combinations of raising and lowering operators 
(i.e., root vectors of $\g_\C$ with respect to its Cartan subalgebra $\h_{\mathbb{C}}$). 
Hence, such $B_j\not\in i\h$ are sums of strictly triangular matrices and have zero diagonal with respect to the basis $\{\ket{v}\}$.
On the other hand, by assumption, $\rho$ is diagonal in this basis; therefore, $\Tr[B_j^\dagger \rho] = 0$ for $B_j \not\in i\h$.
Then Eq.~\eqref{eq:Hg} reduces to a sum over the basis of $i\h$:
\begin{equation}\label{eq:g-purity2}
\rho_{\g}=\sum_{j=1}^{\dim( \h)} \Tr[H_j\ad \rho] H_j = \rho_\h.
\end{equation}
Now recall that $H_j\ad = H_j$, and compute using Eqs.~(26), (27):
\begin{equation}\label{def:weight-basis2}
H_j \rho = \sum_v r_v H_j\ket{v}\bra{v} 
= \sum_v r_v \lambda_v(H_j) \ket{v}\bra{v} ,
\end{equation} 
from where
\begin{equation}\label{def:weight-basis3}
\Tr[H_j \rho] = \sum_v r_v \lambda_v(H_j) = \lambda_\rho(H_j).
\end{equation}
Plugging this into Eq.~\eqref{eq:g-purity2} proves Eq.~(28).

\section{Proof of Corollary 1 }
\setcounter{corollary}{0}

Here we present the derivation of Corollary 1, which we now recall. 

\begin{corollary} \label{cor:concentration-SI}
Let $\norm{O}_2^2\leq 2^n$. If either $\dim(\g)$, $1/\PC_{\g}(\rho)$ or $1/\PC_{\g}(O)$ is in $\Omega(b^n)$ with $b> 2$, then the loss has a BP. 
\end{corollary}

\begin{proof}\let\qed\relax
    Let us first consider the case where $\dim(\g)\in\Omega(b^n)$. Then, irrespective of whether  $O$ or $\rho$ are in  $i\g$, we will have  $\PC_{\g}(O)\leq\Tr[O^2]=\norm{O}_2^2\leq 2^n$ and $\PC_{\g}(\rho)\leq 1$. Hence, 
    \begin{align}
        \Var_{\thv}[\ell_{\thv}(\rho,O)] &= \frac{1}{\dim (\g)} \Tr [O_{\g}^2] \Tr [\rho_{\g}^2]\nonumber\\
        &\leq \frac{1}{\dim (\g)} 2^n\in \OC\left(\frac{1}{c^n}\right)\,,
    \end{align}
    with $c=b/2>1$, and hence the loss has a BP.

    Second, let us analyze  the case where $1/\PC_{\g}(\rho)$ is in $\Omega(b^n)$. Irrespective of the size of the DLA, we will have $\frac{\Tr [O_{\g}^2]}{\dim (\g)}\leq 2^n$. Thus, we have $\Var_{\thv}[\ell_{\thv}(\rho,O)]\in \OC(\frac{1}{c^n})$ with $c=b/2>1$. 
    
    Finally,  if $1/\PC_{\g}(O)$ is in $\Omega(b^n)$, we have $\frac{\Tr [\rho_{\g}^2]}{\dim (\g)}\leq 1$, thus $\Var_{\thv}[\ell_{\thv}(\rho,O)]\in \OC(\frac{1}{b^n})$ with $b>2$.
\end{proof}

\section{Proof of Theorem 2}

In this section, we present a derivation of Theorem 2. We will frequently be working with the Schatten $p$-norms on the spaces of operators $\gl(\calH^{\otimes t})$, where $p=1,\infty$. Both are spectral norms: given the singular values $s_i(A)\geq 0$ of $A$, we may express these norms as
\begin{equation} \label{def:schatten 1 and infty norms 2}
    \norm{A}_1 = \sum_{i} s_i(A), \qquad \norm{A}_\infty = \max_i s_i(A) . 
\end{equation}

This induces corresponding operator norms for linear maps $\calM\colon\gl(\calH^{\otimes t}) \to \gl(\calH^{\otimes t})$ by the usual
\begin{equation} \label{def:induced Schatten p norm 2}
    \norm{\calM}_p = \sup_{\substack{A\neq 0\\ A\in \gl(\calH^{\otimes t})}} \frac{\norm{\calM(A)}_p}{\norm{A}_p}.
\end{equation}

Theorem 2 is the $t=2$ case of the following theorem:

\begin{theorem}\label{th:design-layers-SI}
The ensemble of unitaries $\calE_L\subseteq G$ generated by an $L$-layered circuit  $U(\thv)$ will form an $\epsilon$-approximate $G$ $t$-design, when the number of layers $L$ is
\begin{equation}\label{eq:bound2 on L 2}
    L \geq \frac{\log(1/\epsilon)}{\log\paran{1/\norm{\calA^{(t)}_{\calE_1}}_\infty}} ,
\end{equation} where $\calE_1$ is the ensemble generated by a single layer $U_1(\thv_1)$.
\end{theorem}
Note that to the authors' knowledge, this choice of $L$ need not be uniform in $t$.

\begin{proof}\let\qed\relax
First, recall Eq.~\eqref{eqn:t moment operator}, the definition of the $t$-moment operator $\calM^{(t)}_{\calE_L}\colon\gl(\calH^{\otimes t})\to\gl(\calH^{\otimes t}) $ corresponding to an ensemble $\calE_L$
    \begin{equation} \label{def:moment-op}
        \calM_{\calE_L}^{(t)}(A) = \int_{\calE_L} d\mu(U) \, U^{\otimes t} A (U^\dagger)^{\otimes t}\,, \qquad A\in \gl(\calH^{\otimes t}).
    \end{equation}
Also recall that $\calM_{\calE_L}^{(t)}$ is self-adjoint with respect to the Hilbert--Schmidt inner product on $\gl(\HC^{\otimes t})$, and so $\calM_{\calE_L}^{(t)}$ is diagonalizable on $\gl(\HC^{\otimes t})$ with real eigenvalues.

Since ${\calE_L}\subset G$, all elements of ${\calE_L}$ fix the $G$-invariants $\gl(\HC^{\otimes t})^G$. Hence, as in the proofs of Lemmas~\ref{lem:variance-g-abelian2} and \ref{lem:moments invariance}, the operator $\calM_{{\calE_L}}^{(t)}$ acts as the identity on $\gl(\HC^{\otimes t})^G$. Since it is self-adjoint, the orthogonal complement $\mathfrak{p}$ is also an invariant subspace, and so we may write the orthogonal direct sum
\begin{equation}
\gl(\HC^{\otimes t}) =  \gl(\HC^{\otimes t})^G \oplus \mathfrak{p}.
\end{equation}
Recall from Lemma~\ref{lem:moments invariance} that $\calM_{G}^{(t)}$ is the orthogonal projection onto $\gl(\HC^{\otimes t})^G$.
Thus, $\calA_{{\calE_L}}^{(t)} = \calM_{\calE_L}^{(t)} - \calM_{G}^{(t)}$ acts as zero on $\gl(\HC^{\otimes t})^G$ and as 
$\calM_{{\calE_L}}^{(t)}$ on $\mathfrak{p}$.

Now we investigate the eigenvalues of $\calM_{{\calE_L}}^{(t)}$ on $\mathfrak{p}$.

\begin{lemma} \label{lem:spectral radius of moment operator is in unit disc}
    Every eigenvalue $\lambda$ of $\calM_{{\calE_L}}^{(t)}$ has $\abs{\lambda}\leq 1$.
\end{lemma}
\begin{proof}\let\qed\relax
Since $\calM_{{\calE_L}}^{(t)}$ is self-adjoint, its singular values are equal to the absolute values of its eigenvalues, and so every eigenvalue has $\abs{\lambda}\leq \norm{\calM_{{\calE_L}}^{(t)}}_\infty$ by Eq. (\ref{def:schatten 1 and infty norms 2}). Let $A\in \gl(\calH^{\otimes t})$ have $\norm{A}_\infty = 1$. Then we have
\[
       \norm{\calM_{{\calE_L}}^{(t)}(A)}_\infty \leq \int_{\calE_L} d\mu(U) \; \norm{ U^{\otimes t} A (U^\dagger)^{\otimes t} } _\infty 
       = \int_{\calE_L} d\mu(U) \; 1 = 1,
\]
where we used that the Schatten norms are unitarily invariant. 
\end{proof}

We want to show that all eigenvalues $\lambda$ of $\calM_{{\calE_L}}^{(t)}$ on $\mathfrak{p}$ have $\abs{\lambda}< 1$.
We prove that first in the case of a single layer.

\begin{lemma}\label{lem:single layer ensemble}
    The single-layer ensemble $\calE_1 := \{U_1(\thv_1): \thv_1\in \mathbb{R}^K\}$ with measure $d\mu_1$ satisfies $\abs{\lambda}< 1$ for 
    all eigenvalues $\lambda$ of $\calM_{\calE_1}^{(t)}$ on $\mathfrak{p}$.
\end{lemma}
\begin{proof}\let\qed\relax
    Let $A\in \mathfrak{p}$ with $\norm{A}_\infty=1$. Since $A$ is not $G$-invariant, there exists a $\tilde{U}\in G$ such that for some $\epsilon>0$,
    \[
    \abs{\inprod{A, \tilde{U}^{\otimes t} A (\tilde{U}^\dagger)^{\otimes t}}} < 1-\epsilon .
    \] Note that Corollary 3.2.6 of~\cite{dalessandro2010introduction} assures us that since $G$ compact, there exists a uniform choice of depth $\tilde{L}$ such that any $\tilde{U}\in G$ can be written as
    \begin{equation}\label{eqn:uniform choice in generators}
        \tilde{U}= U_{\tilde{L}}(\tilde{\thv}_{\tilde{L}}) U_{\tilde{L}-1}(\tilde{\thv}_{\tilde{L}-1}) \cdots U_{1}(\tilde{\thv}_{1}).
    \end{equation} Use $(\tilde{{\calE_L}}, d\tilde{\mu})$ to denote the ensemble generated by the depth $\tilde{L}$ circuit. Since the map $\tilde{U}\mapsto \inprod{A, \tilde{U}^{\otimes t}(A)(\tilde{U}^\dagger)^{\otimes t}}$ is a smooth map from $G\to \C$, there exists an open neighborhood $N\subseteq G$ of $\tilde{U}$, on which
    \[
        \abs{\inprod{A,V^{\otimes t} A (V^\dagger)^{\otimes t}}}< 1-\frac{\epsilon}2  \qquad \text{for all }\; A\in N.
    \] As the choice of depth $\tilde{L}$ is uniform for all elements in $G$, Eq.~(\ref{eqn:uniform choice in generators}) assures us that $N\cap \supp(d\tilde{\mu})$ is not measure 0. This then means that 
    \begin{align*}
        \abs{\inprod{A, \calM_{\tilde{{\calE_L}}}^{(t)} (A)}} < 1.
    \end{align*} Thus any eigenvalue $\lambda$ of $\calM_{\tilde{{\calE_L}}}^{(t)}\big\vert_{\mathfrak{p}}$ must have $\abs{\lambda}<1$. But a straightforward calculation reveals that $\calM_{\tilde{{\calE_L}}}^{(t)} = (\calM_{U_1}^{(t)})^{\tilde{L}}$, and so any eigenvalue $\lambda$ of $\calM_{U_1}^{(t)}\big\vert_{\mathfrak{p}}$ has $\abs{\lambda}<1$.
\end{proof}

Now to finish the proof of the theorem, we combine the results of Lemmas~\ref{lem:spectral radius of moment operator is in unit disc} and \ref{lem:single layer ensemble}. It follows that for the single layer ensemble $\calE_1$, the operator $\calA_{\calE_1}^{(t)} = \calM_{G}^{(t)} - \calM_{\calE_1}^{(t)}$ is diagonalizable with all eigenvalues $\lambda$ satisfying $\abs{\lambda}<1$. Therefore, its operator norm $\norm{\calA_{\calE_1}^{(t)}}_\infty<1$.
For a circuit of depth $L$, we may use the following Lemma~\ref{lem:L layer to single layer} to see that $\norm{\calA_{{\calE_L}}^{(t)}}_\infty =  \paran{\norm{\calA_{\calE_1}^{(t)}}_\infty}^L$, so we can pick $L$ as in Eq.~\eqref{eq:bound2 on L 2} to grant our result.
\end{proof}

\begin{lemma} \label{lem:L layer to single layer}
    Let $\calE_L\subseteq G$ be the ensemble of unitaries generated by an $L$-layered circuit $U(\thv)$. Then 
    \begin{equation}
        \calA_{\calE_L}^{(t)} = \underbrace{\calA_{\calE_1}^{(t)}\circ \calA_{\calE_1}^{(t)} \circ \dots \circ \calA_{\calE_1}^{(t)}}_{L \text{ times}} . 
    \end{equation} In particular, $\norm{\calA_{\calE_L}^{(t)}}_\infty = \norm{\calA_{\calE_1}^{(t)}}_\infty^L$.
\end{lemma}
\begin{proof}\let\qed\relax
    Let us more explicitly write $\calA_{\calE_L}^{(t)}$:
\begin{equation}
    \calA_{\calE_L}^{(t)}(\cdot)=\calM_{\calE_L}^{(t)}(\cdot)- \calM_G^{(t)}(\cdot) = \int_{\calE_L} d\mu(U(\thv)) \, U(\thv)^{\otimes t} (\cdot) (U(\thv)^\dagger)^{\otimes t}-\int_G d\mu(U(\thv)) \, U(\thv)^{\otimes t} (\cdot) (U(\thv)^\dagger)^{\otimes t}\,.
\end{equation}
Taking the composition of two such linear maps leads to 
\begin{equation}
\calA_{\calE_L}^{(t)}\circ\calA_{\calE_L}^{(t)}=\calM_{\calE_L}^{(t)}\circ\calM_{\calE_L}^{(t)}+\calM_{G}^{(t)}\circ\calM_{G}^{(t)}-\calM_{G}^{(t)}\circ\calM_{\calE_L}^{(t)}-\calM_{\calE_L}^{(t)}\circ\calM_{G}^{(t)}\,.
\end{equation}
Then, the following hold
\begin{equation}
\calM_{\calE_L}^{(t)}\circ\calM_{\calE_L}^{(t)}=\calM_{\calE_{2L}}^{(t)}\,,\quad\quad  \calM_{G}^{(t)}\circ\calM_{G}^{(t)}=\calM_{G}^{(t)}\circ\calM_{\calE_L}^{(t)}=\calM_{\calE_L}^{(t)}\circ\calM_{G}^{(t)}=\calM_{G}^{(t)}\,.
\end{equation}
The first equality simply states that the multiplication of two distributions from $L$-layered circuits is equal to the distribution of a $2L$-layered circuit, while the latter ones follows from the left- and right-invariance of the Haar measure. Hence, we have 
\begin{equation}
\calA_{\calE_L}^{(t)}\circ\calA_{\calE_L}^{(t)}=\calA_{\calE_{2L}}^{(t)}\,.
\end{equation}
In particular, we obtain
\begin{equation}\label{eq:comp-maps}
\calA_{\calE_L}^{(t)}=\underbrace{\calA_{\calE_1}^{(t)}\circ \calA_{\calE_1}^{(t)} \circ \dots \circ \calA_{\calE_1}^{(t)}}_{L \text{ times}}\,,
\end{equation}
where $\calA_{\calE_1}^{(t)}$ corresponds to the single layer ensemble $\calE_1$. 

The final claim that $\norm{\calA_{\calE_L}^{(t)}}_\infty = \norm{\calA_{\calE_1}^{(t)}}_\infty^L$ follows by noting that $\calA_{\calE_1}^{(t)}$ is self-adjoint, and so its singular values are exactly the absolute values of its eigenvalues. Take the largest singular value $\abs{\lambda}$ and a corresponding eigenvector $A$ of $\calA_{\calE_1}^{(t)}$. Then by Eq. (\ref{eq:comp-maps}), $A$ is an eigenvector of $\calA_{\calE_L}^{(t)}$ of largest singular value $\abs{\lambda}^L$.
\end{proof}

\subsection{Example of a use-case of Theorem 2}

In this section we present an example where  we can use Theorem 2 to predict the number of layers needed for an $L$-layered circuit to form an $\epsilon$-approximate $2$-design.

Let $S$ be  a subgroup of  $\mathbb{U}(d)$,  the unitary group of degree $d=2^n$ 
acting on $\HC$, and let $\{P_j\}_{j=1}^D$ be a basis for the commutant of the $t$-fold  tensor representation $\HC^{\otimes t}$ of $S$. That is, $[P_j,U^{\otimes t}]=0$ for all $U\in S$. Then, we know that (see, e.g., \cite{mele2023introduction}):
\begin{equation} \label{def:moment-op-vecto}
        \mathcal{M}_{S}^{(t)} = \int_{S} d\mu(U) \, (U\otimes U^*)^{\otimes t}=\sum_{i,j=1}^D W^{-1}_{ij}|P_i\rangle\rangle\langle\langle P_j|\,. 
\end{equation}
Here $W^{-1}$ is the so-called Weingarten matrix, the inverse of the Gram matrix $W$ whose entries $W_{ij}=\Tr[P_i\ad P_j]$ are given by the Hilbert--Schmidt inner product of the commutant's basis elements. Moreover, we have denoted as $|A\rangle\rangle$ the standard vectorization of a linear operator $A$. Namely, given some $A\in\LC(\HC)$ such that $A=\sum_{\mu,\nu=1}^{2^n}A_{\mu\nu}|\mu\rangle\langle \nu|$, then $|A\rangle\rangle =\sum_{\mu,\nu=1}^{2^n}A_{\mu\nu}|\mu\rangle\otimes |\nu\rangle$. Hence, if we can characterize the basis of the commutant of the $t$-fold tensor representation of $S$, then we can find a closed form expression for $\mathcal{M}_{S}^{(t)}$. Finally, note that $\mathcal{M}_{S}^{(t)}$ is an operator acting on $2t$ copies of the Hilbert space $\HC$, and hence is of dimension $2^{2nt}\times 2^{2nt}$.

Next, let us recall from the main text that we have factorized the parametrized quantum circuit into layers as $U(\thv) = U_{L}(\thv_L)  U_{L-1}(\thv_{L-1})  \cdots  U_1(\thv_1)$, and that each layer $U_{l}(\thv_l)$ is given by $
U_l(\thv_l) = \prod_{k=1}^K U_k(\theta_{l,k})$. Hence, if the parameters in each $U_k$ are independently sampled, it is not hard to see that the $t$-th moment operator of a single layer can be expressed as the product of the $t$-th moment operators of each $U_k$. That is, 
\begin{equation}
\mathcal{M}_{\mathcal{E}_1}^{(t)}=\prod_{k=1}^K    \mathcal{M}_{\mathcal{E}_{1,k}}^{(t)}\,,
\end{equation}
where we have denoted as $\mathcal{E}_{1,k}$ the distribution of unitaries obtained from each $U_k(\theta_{l,k})$. As such, if we assume that each $\mathcal{E}_{1,k}$ forms a subgroup, and that we sample it over its respective Haar measure, we will have via Eq.~\eqref{def:moment-op-vecto} that 
\begin{equation}\label{eq:formula}
\mathcal{M}_{\mathcal{E}_1}^{(t)}=\sum_{i_{1},j_{1}=1}^{D_1}\cdots\sum_{i_{K},j_{K}=1}^{D_K} \left((W^{-1}_1)_{i_1j_1}\cdots(W^{-1}_K)_{i_Kj_K}\langle\langle P_{j_1}|P_{i_2}\rangle\rangle\cdots \langle\langle P_{j_{K-1}}|P_{i_K}\rangle\rangle\right)|P_{i_1}\rangle\rangle\langle\langle P_{j_K}|\,.
\end{equation}
Here, we denoted as $\{P_{j_k}\}_{j_k=1}^{D_k}$ be a basis for the commutant of the $t$-fold  tensor representation of $\mathcal{E}_{1,k}$, and as $W^{-1}_k$ its associated Weingarten matrix. Moreover, we recall that $\langle\langle A|B\rangle\rangle=\Tr[A\ad B]$.  Equation~\eqref{eq:formula} indicates that we can compute the single layer operator $\mathcal{M}_{\mathcal{E}_1}^{(t)}$ by decomposing the layer into unitaries whose ensembles form groups. In its lower form, such procedure can always be taken to the single-gate level where each   $U_k(\theta_{l,k})=e^{-iH_k \theta_{l,k}}$ with $H_k^2=\id$, and hence where its associated distribution will be a representation of $\mathbb{U}(1)$.

\begin{figure*}
    \centering
    \includegraphics[width=1\linewidth]{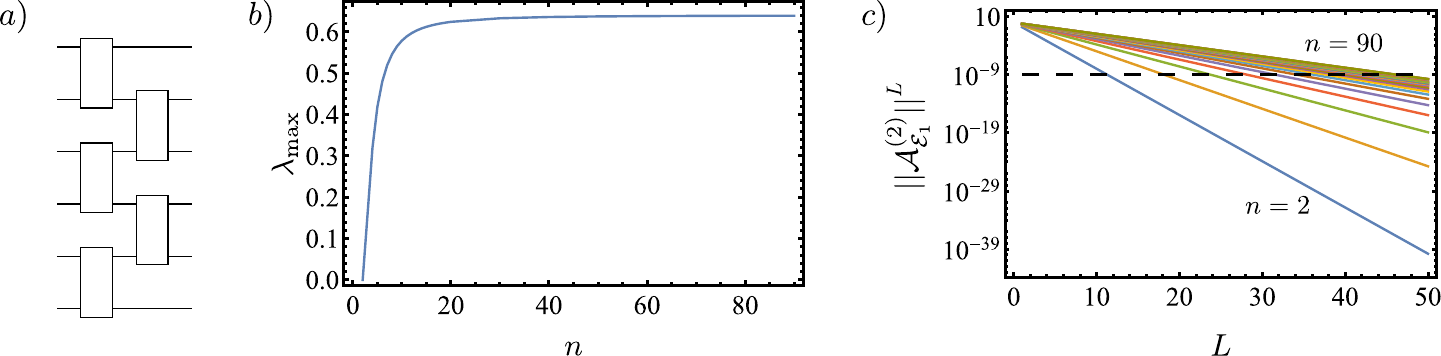}
    \caption{\textbf{Layers needed to become an $\epsilon$-approximate $2$-design.} (a) Hardware efficient ansatz where two-qubit gates act in a brick-like fashion on neighboring pairs of qubits. We assume that each gate is independently and  randomly sampled from $\mathbb{SU}(4)$. (b) Largest eigenvalue of $\mathcal{A}_{\mathcal{E}_1}^{(2)}$ as a function of the number of qubits. (c) Each curve shows $(\mathcal{A}_{\mathcal{E}_1}^{(2)})^L$ for different values of $n$. The thick horizontal dashed line is plotted at $10^{-9}$. Vertical axis is shown in a logarithmic scale.     }
    \label{fig:si}
\end{figure*}

For instance, let us consider a hardware efficient circuit composed of alternating layers of two-qubit gates acting on neighboring qubits (see Fig.~\ref{fig:si}(a)). Moreover, let us assume that each two-qubit gate forms an independent local $2$-design over $\mathbb{SU}(4)$. If we consider a gate acting on qubits $i$ and $i+1$, we will have~\cite{mele2023introduction}
\begin{equation} 
        \mathcal{M}_{\mathbb{SU}(4)}^{(2)} = \frac{1}{15}\left(|\mathbb{I}_i\mathbb{I}_{i+1}\rangle\rangle\langle\langle\mathbb{I}_i\mathbb{I}_{i+1}|+|\mathbb{S}_i\mathbb{S}_{i+1}\rangle\rangle\langle\langle \mathbb{S}_i\mathbb{S}_{i+1}|-\frac{1}{4}\left(\mathbb{S}_i\mathbb{S}_{i+1}|+|\mathbb{S}_i\mathbb{S}_{i+1}\rangle\rangle\langle\langle \mathbb{I}_i\mathbb{I}_{i+1}|\right))\right)\,. 
\end{equation}
where $\mathbb{I}_i$ and $\mathbb{S}_i$ denote the identity and SWAP operators acting on the two copies of the $i$-th qubit Hilbert spaces. By defining a basis for the operator space spanned from the non-orthogonal  elements 
\begin{align}
\bigl\{|\mathbb{I}_i\mathbb{I}_{i+1}\rangle\rangle,|\mathbb{I}_i\mathbb{S}_{i+1}\rangle\rangle,|\mathbb{S}_i\mathbb{I}_{i+1}\rangle\rangle,|\mathbb{S}_i\mathbb{S}_{i+1}\rangle\rangle\bigr\},
\end{align}
we can express $\mathcal{M}_{\mathbb{SU}(4)}^{(2)}$ as a $4\times 4$ matrix
\begin{equation}
    \mathcal{M}_{\mathbb{SU}(4)}^{(2)}=
    \begin{pmatrix}
        1&0&0&0\\
        2/5 & 0 & 0&2/5\\
        2/5 & 0 & 0&2/5\\
        2/5 & 0 & 0&2/5\\
        0 &  0 & 0 & 1
    \end{pmatrix}\,.
\end{equation}
One can readily verify that $\mathcal{M}_{\mathbb{SU}(4)}^{(2)}\circ \mathcal{M}_{\mathbb{SU}(4)}^{(2)}= \mathcal{M}_{\mathbb{SU}(4)}^{(2)}$ and that $\Tr[\mathcal{M}_{\mathbb{SU}(4)}^{(2)}]=2$, as expected from the fact that it is a projector onto a subspace of dimension $2$. Assuming for simplicity that $n$ is even (the  calculations below can be readily generalized to the case when $n$ is odd), we find that the single layer moment operator is
\begin{equation}\label{eq:moment-op-reduced}
    \mathcal{M}_{\mathcal{E}_1}^{(t)}=\left(\id\otimes \left(\mathcal{M}_{\mathbb{SU}(4)}^{(2)}\right)^{\otimes (n/2-1)}\otimes \id\right)\circ \left(\mathcal{M}_{\mathbb{SU}(4)}^{(2)}\right)^{\otimes (n/2)}\,,
\end{equation}
where  $\id$ denotes the $2\times 2$ identity matrix. Note that the previous has allowed us to effectively express $\mathcal{M}_{\mathcal{E}_1}^{(t)}$ as a $2^n\times 2^n$ matrix, which is much smaller than its vanilla form of $16^{n}\times 16^{n}$. Moreover, this approach follows very closely manuscripts in the literature where Haar average properties are computed by mapping objects such as $\mathcal{M}_{\mathcal{E}_1}^{(t)}$ to Markov-like mixing process matrices~\cite{harrow2018approximate,dalzell2022randomquantum,napp2022quantifying,braccia2024computing}.

To compute the norm of the operator $\mathcal{A}_{\mathcal{E}_1}^{(2)}=\mathcal{M}_{\mathcal{E}_1}^{(2)}- \mathcal{M}_G^{(2)}$ we can use the fact that deep hardware efficient ansatz are universal~\cite{larocca2021diagnosing}, meaning  that  $\mathfrak{g}=\mathfrak{su}(d)$, and hence $G=\mathbb{SU}(d)$. From here, we can express 
\begin{equation}
    \mathcal{M}_{\mathbb{SU}(d)}^{(2)} = \frac{1}{4^n-1}\left(|\mathbb{I}\rangle\rangle\langle\langle\mathbb{I}|+|\mathbb{S}\rangle\rangle\langle\langle \mathbb{S}|-\frac{1}{2^n}\left(|\mathbb{I}\rangle\rangle\langle\langle\mathbb{S}|+|\mathbb{S}\rangle\rangle\langle\langle \mathbb{I}|\right))\right)\,,
\end{equation} 
where $|\mathbb{I}\rangle\rangle=\bigotimes_{i=1}^n |\mathbb{I}_i\rangle\rangle $ and $|\mathbb{S}\rangle\rangle=\bigotimes_{i=1}^n |\mathbb{S}_i\rangle\rangle $. Following a similar procedure such as that used to build an efficient representation of $\mathcal{M}_{\mathcal{E}_1}^{(t)}$ in Eq.~\eqref{eq:moment-op-reduced}, we can express $\mathcal{M}_{\mathbb{SU}(d)}^{(2)}$ and thus $\mathcal{A}_{\mathcal{E}_1}^{(2)}$ as a $2^n\times 2^n$ matrix. Combining this with the fact that $\mathcal{A}_{\mathcal{E}_1}^{(2)}$ can be represented efficiently as a matrix product operator, and the fact that the eigenvector associated to its largest eigenvalue admits an efficient representation as a bond dimension 4 matrix product state, we can easily obtain the largest eigenvalue for $\mathcal{A}_{\mathcal{E}_1}^{(2)}$. 

For instance, in Fig.~\ref{fig:si}(b) we plot the largest eigenvalue $\lambda_{\max}$ of $\mathcal{A}_{\mathcal{E}_1}^{(2)}$ as a function of the number of qubits $n$ that $U(\thv)$ acts on. We have considered system sizes $n=2,\ldots,90$.  Here we can see that for $n=2$, $\lambda_{\max}=0$, as expected. Moreover, as $n$ increases, the value of $\lambda_{\max}$ appears to saturate at a value $\approx 0.639$. Then, in Fig.~\ref{fig:si}(c) we plot $\norm{\mathcal{A}_{{\mathcal{E}_L}}^{(2)}} = \paran{\norm{\mathcal{A}_{\mathcal{E}_1}^{(2)}}}^L$ for different values of $L$. Here, we find that as the number of layers increases, the norm of the expressiveness superoperator $\norm{\mathcal{A}_{{\mathcal{E}_L}}^{(2)}}$ decreases exponentially. Indeed, we can now answer questions such as: ``\textit{How many layers are needed for the circuit to become an $\epsilon=10^{-9}$ approximate design?}'' For the number of qubits considered, we show that $L=47$ will suffice.

To finish, we note that here we have considered a circuit whose local gates are randomly taken from $\mathbb{SU}(4)$. However, the methods presented can be extended and generalized to other circuit architectures where the gates are sampled from any subgroup. In particular, we can always assume that each gate in the layer is generated by a Pauli and hence samples from a representation of $\mathbb{U}(1)$. We leave the exploration of such cases for future work.

\section{Proof of Theorem 3}

Here we present a proof for Theorem 3 of the main text, which we recall for convenience:

\begin{theorem}\label{th:variance-layers-sm}
Let $\mathfrak{g}\subseteq \mathfrak{u}(2^n)$ be any dynamical Lie algebra, and suppose either $\rho$ or $O$ are in $i\g$ with $\rho$ a density matrix. Then the difference between the loss function variance of an $L$-layered circuit, with ensemble of unitaries $\calE_L$, and the variance for a circuit that forms a $2$-design over $G$ can be bounded as 
\begin{equation}
\abs{\Var_{\calE_L} [\ell_{\thv}(\rho, O)]-\Var_{G} [\ell_{\thv}(\rho, O)]}\leq 3 \paran{\norm{\calA_{\calE_1}^{(2)}}_\infty}^L \norm{O}_1^2 \,,
\end{equation}
where the usual Schatten $p$-norm and induced Schatten $p$-norm are given by Eq. (\ref{def:schatten 1 and infty norms 2}) and Eq. (\ref{def:induced Schatten p norm 2}).
\end{theorem}

\begin{proof}\let\qed\relax
Let us begin by observing that we may split
\begin{equation}\begin{split} \label{eq:var estimate ineq}
    \abs{\Var_{\calE_L}[\ell_{\thv}(\rho,O)] - \Var_{G}[\ell_{\thv}(\rho,O)]} &= \abs{\bbE_{\calE_L} [\ell_{\thv}(\rho,O)]^2 - (\bbE_{\calE_L} [\ell_{\thv}(\rho,O)])^2 + \bbE_{G} [\ell_{\thv}(\rho,O)]^2 - (\bbE_{G} [\ell_{\thv}(\rho,O)])^2 } \\
    &\leq \abs{\bbE_{\calE_L} [\ell_{\thv}(\rho,O)]^2 - \bbE_{G} [\ell_{\thv}(\rho,O)]^2 } + \abs{(\bbE_{\calE_L} [\ell_{\thv}(\rho,O)])^2 - (\bbE_{G} [\ell_{\thv}(\rho,O)])^2 } . 
\end{split}\end{equation}
To estimate the first term of Eq. (\ref{eq:var estimate ineq}), we mimic Eq. (\ref{eqn:variance2}) to write the second moments in terms of the second moment operators of ensembles $\calE$, which again are linear maps $\calM_\calE^{(2)}\colon \gl(\calH)^{\otimes 2} \to \gl(\calH)^{\otimes 2}$,
\begin{equation}\begin{split} \label{eq:estimating second moment of L-layer circuit}
    \abs{\bbE_{\calE_L} [\ell_{\thv}(\rho,O)]^2 - \bbE_{G} [\ell_{\thv}(\rho,O)]^2 } &= \abs{ \Tr \left[\calM_{\calE_L}^{(2)}(\rho^{\otimes 2}) O^{\otimes 2}\right] - \Tr \left[\calM_{G}^{(2)}(\rho^{\otimes 2}) O^{\otimes 2}\right] } \\
    &= \abs{ \Tr\left[ (\calM_{\calE_L}^{(2)} - \calM_G^{(2)})(\rho^{\otimes 2}) O^{\otimes 2} \right]} \\
    &\leq \norm{\rho^{\otimes 2}}_1 \norm{O^{\otimes 2}}_1 \norm{\calA^{(2)}_{\calE_L}}_\infty \\
    &\leq \norm{O^{\otimes 2}}_1 \norm{\calA^{(2)}_{\calE_L}}_\infty,
\end{split}\end{equation} where we used H\"older's inequality for matrices and that $\norm{\rho^{\otimes 2}}_\infty \leq 1$ since $\rho$ a density matrix and so has eigenvalues at most 1.

To estimate the second term of Eq.~(\ref{eq:var estimate ineq}), we use the difference of squares formula $a^2-b^2 = (a+b)(a-b)$ to write
\begin{equation}\begin{split}
    \abs{(\bbE_{\calE_L} [\ell_{\thv}(\rho,O)])^2 - (\bbE_{G} [\ell_{\thv}(\rho,O)])^2 } = \abs{\bbE_{\calE_L} [\ell_{\thv}(\rho,O)] + \bbE_{G} [\ell_{\thv}(\rho,O)]}\; \cdot \;  \abs{\bbE_{\calE_L} [\ell_{\thv}(\rho,O)] - \bbE_{G} [\ell_{\thv}(\rho,O)]} . 
\end{split}\end{equation} Recalling that for any ensemble $\calE$ the first moment operator $\calM_\calE^{(1)}\colon \gl(\calH)\to \gl(\calH)$ has operator norm $\norm{\calM_\calE^{(1)}}_\infty \leq 1$ by Lemma \ref{lem:spectral radius of moment operator is in unit disc}, we see that
\begin{equation}
    \abs{\bbE_{\calE_L} [\ell_{\thv}(\rho,O)] + \bbE_{G} [\ell_{\thv}(\rho,O)]} = \abs{\Tr \left[ (\calM_{\calE_L}^{(1)} + \calM_G^{(1)})(\rho) O \right]}
    \leq \norm{O}_1 \norm{\calM_{\calE_L}^{(1)}+\calM_{G}^{(1)}}_\infty 
    \leq 2 \norm{O}_1 , 
\end{equation} again by H\"older's inequality and $\norm{\rho}_\infty < 1$. Essentially the same estimate as for the variance shows that
\begin{equation}
    \abs{\bbE_{\calE_L} [\ell_{\thv}(\rho,O)] - \bbE_{G} [\ell_{\thv}(\rho,O)]} \leq \norm{O}_1 \norm{\calA_{\calE_L}^{(1)}}_\infty . 
\end{equation}

Putting it all together into Eq. (\ref{eq:var estimate ineq}) and using that $\norm{O}_1^2 = \norm{O^{\otimes 2}}_1$, we have that

\begin{equation}\begin{split}
    \abs{\Var_{\calE_L}[\ell_{\thv}(\rho,O)] - \Var_{G}[\ell_{\thv}(\rho,O)]} &\leq \norm{O}_1^2 \paran{\norm{\calA_{\calE_L}^{(2)}}_\infty + 2\norm{\calA_{\calE_L}^{(1)}}_\infty} \\
    &= \norm{O}_1^2 \paran{\norm{\calA_{\calE_1}^{(2)}}_\infty^L + 2\norm{\calA_{\calE_1}^{(1)}}_\infty^L} \\
    &\leq 3\norm{O}_1^2 \norm{\calA_{\calE_1}^{(2)}}_\infty^L , 
\end{split}\end{equation} where in the second line we have used Lemma~\ref{lem:L layer to single layer}, and in the third line we have used Lemma \ref{lem:every t design is a t-1 design} below to upper bound $\norm{\calA_{\calE_1}^{(1)}}_\infty\leq \norm{\calA_{\calE_1}^{(2)}}_\infty$.
\end{proof}

\begin{lemma} \label{lem:every t design is a t-1 design}
    Let $\calE$ be an ensemble of unitaries that forms an $\epsilon$-approximate $G$ $t$-design, i.e., $\norm{\calM_\calE^{(t)} - \calM_G^{(t)}}_\infty <\epsilon $. Then $\calE$ forms an $\epsilon$-approximate $s$-design for all $1\leq s\leq t$. 
\end{lemma}
\begin{proof}\let\qed\relax
    We show that $\calE$ forms an $\epsilon$-approximate $(t-1)$-design, whence the other cases follow. Let $\calM_\calE^{(t)}, \calM_G^{(t)}\colon \gl(\calH^{\otimes t})\to \gl(\calH^{\otimes t})$. We may first note that for any $\rho \in \gl(\calH^{\otimes (t-1)})$, we have
    \begin{equation}
        \calM_{\calE}^{(t)}( \rho\otimes \mathds{1}_\calH) = \int_\calE U^{\otimes (t-1)} \rho (U^\dagger)^{\otimes (t-1)} \otimes U U^\dagger \, d\mu = \paran{\int_\calE U^{\otimes (t-1)} \rho (U^\dagger)^{\otimes (t-1)} \, d\mu } \otimes \mathds{1}_\calH = \calM_{\calE}^{(t-1)}(\rho)\otimes \mathds{1}_\calH,
    \end{equation} and likewise for the Haar ensemble $G$. Then computing the norm, we obtain
    \begin{equation}\begin{split}
        \norm{\calM_{\calE}^{(t-1)} - \calM_G^{(t-1)} }_\infty = \sup_{\substack{\rho\in \gl(\calH^{\otimes (t-1)}) \\ \rho\neq 0}} \frac{\norm{ (\calM_{\calE}^{(t-1)} - \calM_G^{(t-1)})(\rho) }_\infty}{\norm{\rho}_\infty} &= \sup_{\substack{\rho\in \gl(\calH^{\otimes (t-1)}) \\ \rho\neq 0}} \frac{\norm{ (\calM_{\calE}^{(t)} - \calM_G^{(t)})(\rho\otimes \mathds{1}) }_\infty}{\norm{\rho\otimes \mathds{1}}_\infty} \\
        &\leq \norm{\calM_\calE^{(t)} - \calM_G^{(t)}}_\infty \\
        &<\epsilon .
    \end{split}\end{equation}
\end{proof}

\begin{remark}
    While we have expressed Theorem~\ref{th:variance-layers-sm} with the intent of discussing the variance, one may wonder about bounds quantifying the convergence of deep parameterized quantum circuits to the Haar measure for higher moments. It is not difficult to see that mimicking Eq. (\ref{eq:estimating second moment of L-layer circuit}) for the $t^{th}$ moment yields 
    \begin{equation}
        \abs{\bbE_{\calE_L} [\ell_{\thv}(\rho,O)]^t - \bbE_{G} [\ell_{\thv}(\rho,O)]^t } \leq \norm{O^{\otimes t}}_1 \norm{\calA_{\calE_1}^{(t)}}_\infty^L .
    \end{equation} Then, to discuss quantities such as the $t^{th}$ cumulant, Lemma \ref{lem:every t design is a t-1 design} guarantees that for any ensemble $\calE_L$ and any integer $1\leq s\leq t$, we have $\norm{\calA_{\calE_L}^{(s)}}_\infty \leq \norm{\calA_{\calE_L}^{(t)}}_\infty $, which grants control over lower-order moments in terms of the $t^{th}$ moment.
    
\end{remark}

\section{Numerical simulations}
In this section, we present numerical experiments that (i) assess the exactness of Eq.~(9) and (ii) probe some non-intuitive aspects revealed by the Lie-algebraic perspective. 

In the following, we consider the trainability of circuits composed of parameterized single-qubit $Z$-rotations and parameterized two-qubit $XX$-rotations acting over nearest qubit neighbors arranged in a one-dimensional topology.
This accounts for a total of $(2n-1)$ parametrized gates per layer of the parametrized quantum circuit.
The DLA associated to this circuit is
\begin{equation}
    \mathfrak{g}=\left\langle 
    \left( \{ i X_j X_{j+1} \}^{n-1}_{j=1} \right)
    \cup \left( \{ i Z_j \}^n_{j=1} \right)\right\rangle_{\operatorname{Lie}}.
    \label{eqn:g_tfim}
\end{equation}
As shown in~\cite{kokcu2022fixed}, the DLA is a simple Lie algebra $\g\cong\mathfrak{so}(2n)$ that has dimension ${\rm dim}(\g) = n (2n-1)\in\Theta(\poly(n))$. 
An orthogonal basis of $\g$ is given by
\begin{equation}
    i\{ \widehat{X_iX_j}, \widehat{X_iY_j}, \widehat{Y_iX_j}, \widehat{Y_iY_j} \}_{1 \leq i < j \leq n}, \quad \text{where}\quad \widehat{A_iB_j} := A_i Z_{i+1}\cdots Z_{j-1} B_j. 
    \label{eqn:g_tfim_basis}
\end{equation}
Given the Cartan subalgebra $\h= {\rm span}_{\mathbb{R}} \{ Z_i\}^{n}_{i=1}$ of $\g$, one readily identifies the state $\ket{hw} := \ket{0}^{\otimes n}$ as one highest weight vector of the algebra with $\g$-purity $\PC_{\g}(\ket{hw}\bra{hw})=\dim \h / 2^n = n/ 2^n.$  

In the simulations, we study the variance of the loss for varied system sizes $n\in [3, 15]$, and consider 4 distinct setups each characterized by the choice of the measurement observable $O$ and of the pure initial state $|\psi_0\rangle$. 
In all cases $O\in i \g$. 
\begin{itemize}
\item \textbf{Setup 0:}  $O = X_pX_{p+1} + Z_p$ (with $p=\lfloor \frac{n}{2} \rfloor$ indexing the middle of the chain) and $|\psi_0\rangle=|hw\rangle$. This illustrates the case where the operator contains a generalized local component and a generalized-nonlocal component; and where the initial state has maximal $\g$-purity.
\item \textbf{Setup 1:} $O = \widehat{X_1Y_n}$ and $|\psi_0\rangle=\ket{hw}$. This illustrates the case where the operator is global  (i.e., it acts non trivially on the whole system) but belongs in $i\g$, and where the initial state has maximal $\g$-purity.
\item \textbf{Setup 2:} $O = X_pX_{p+1} + Z_p$ and $|\psi_0\rangle$ is generated by applying a single layer of random single-qubit rotations to $\ket{hw}$. 
\item \textbf{Setup 3:} $O = X_pX_{p+1} + Z_p$ and $|\psi_0\rangle$ is generated by applying $n$ layers of a random circuit composed of two-qubit rotations to $\ket{hw}$. 
\end{itemize}
As we soon discuss further, these last two setups illustrate cases where the $\g$-purity of the state input to the circuit has been decreased (compared to the original $\ket{hw}$ state), to a different extent, by means of unitary rotations.

For the setups just described, we systematically estimate variances of the cost over a total of $5,000$ random initializations of the circuit parameters.
For each system size, the variances are estimated for increased number of layers, until they converge to a fixed value (i.e., until the circuit forms an approximate design). 
In Fig.~\ref{fig:tfim_free}, we report all these estimated variances (solid lines with distinct colors for each of the setups), and also, the theoretical predictions obtained from Eq.~(9) in Theorem~1. 
Over all the setups studied, theory and numerically estimated variances match closely (up to small errors arising from the sample variance uncertainty). 
Going further, it is verified that measuring a global measurement operator (Setup 1) does not necessarily incur BPs provided that it belong to $i\g$. Furthermore, we can see that under unitary transformation (Setups 2 and 3) the $\g$-purity of the input states are indeed decreased.  Otherwise, these would result in similar variances as Setup 0, as they share the same measurement observable. In particular, we highlight the fact that in Setup 2 we perform single qubit rotations (which do not change the standard purity $\Tr[\rho^2]$, nor the standard entanglement), but which incur in a small decrease of the variance. Note that in the latter the variance decrease is small, as it is still  polynomially vanishing (i.e., this setup does not incur BPs). This result is in contrast to Setup 3, where the circuit applied to $\ket{hw}$ generates enough generalized entanglement to induce a BP, as evidenced by an exponentially decaying variance.

\begin{figure}
    \includegraphics[width=0.4\textwidth]{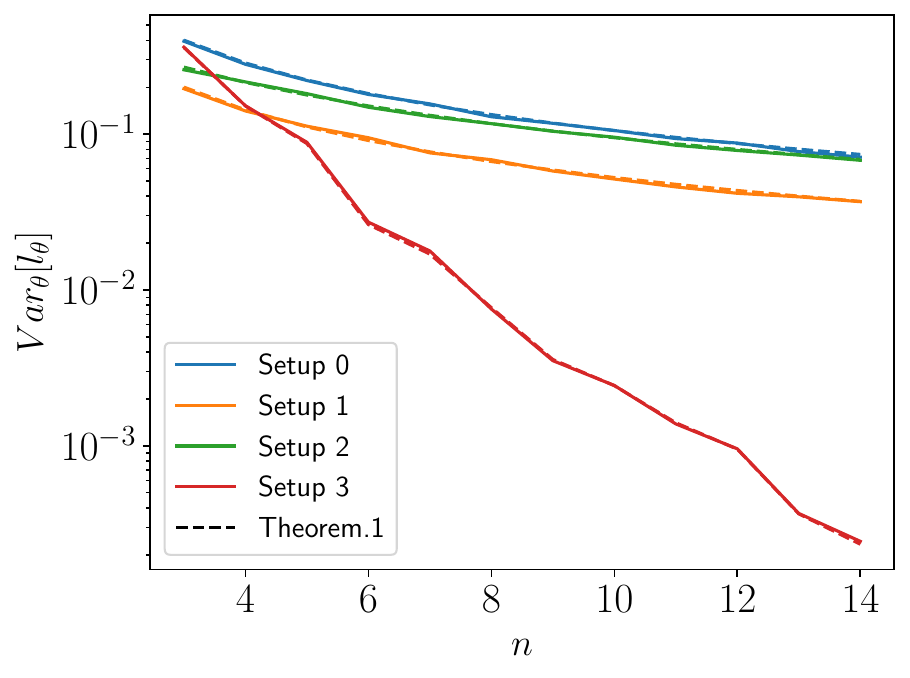}
    \caption{
    \textbf{Scaling of the variance of the loss functions for 4 different settings.} We numerically estimate variances of the loss function for a circuit with a  DLA given in Eq.~\eqref{eqn:g_tfim} for 4 different setups composed of different pairs of measurement observables and initial states. Each setup is  described in the main text,  and the associated variance corresponds to a solid in the plot. These estimated variances are compared to exact evaluation of Eq.~(9) of Theorem~1 (dashed lines).}
    \label{fig:tfim_free}
\end{figure}

\end{document}